%% file: main.tex
\documentclass[acmsmall,nonacm]{acmart}

\input{sty/packages.tex}
\input{sty/defs.tex}
\input{sty/code.tex}

\usepackage{thmtools,thm-restate}
 
\usepackage{algorithmicx}
\usepackage{algorithm} 
\usepackage{algpseudocode}
\usepackage{xifthen}
\usepackage{graphicx}
\usepackage{wrapfig}

\usepackage{stmaryrd}
\usepackage{amsmath}
\usepackage{amsthm}

\newcommand{\pierre}[1]{{\color{blue}{[PC: #1}]}}

\newcommand{\remove}[1]{}

\usepackage{cancel}
\newif\ifcomments
\commentstrue
\newcommand{\add}[1]{\ifcomments {\color{purple}{#1}}\else #1\fi}
\newcommand{\rem}[1]{\ifcomments {\cancel{#1}}\fi}
\newcommand{\replace}[2]{\ifcomments \rem{#1}\add{#2}\else #2\fi}
\usepackage{ulem}

\newcommand{\ms}[1]{%
	    \relax\ifmmode
	        \mathord{\mathcode`\-="702D\it #1\mathcode`\-="2200}%
	    \else
	        {\it #1}%
	    \fi
}
\newcommand{\tup}[1]{%
	    \relax\ifmmode
	      \langle #1 \rangle%
	    \else
	        $\langle$ #1 $\rangle$%
	    \fi
}

\date{}

\setcopyright{none} 
\acmConference[Anonymous Submission to ACM CCS 2021]{ACM Conference on Computer and Communications Security}{Due 06 May 2021}{London, TBD}
\acmYear{2019}

\begin{document}

\normalem

\algsetblockdefx[LocalState]{LocalState}{EndLocalState}{}{}{\textbf{Local state $\alpha_p$:}}{}
\algsetblockdefx[Round]{Round}{EndRound}{}{}[1]{\textbf{Round:} $#1$}{}
\algsetblockdefx[SendStep]{SendStep}{EndSendStep}{}{}{$S_p^r$:}{}
\algsetblockdefx[TransitionStep]{TransitionStep}{EndTransitionStep}{}{}{$T_p^r$:}{}
\algsetblockdefx[Upon]{Upon}{EndUpon}{}{}[2]{\textbf{upon} $\mathit{deliver}(#1)$ from $#2$ \textbf{do}}{}
\algsetblockdefx[UponEvent]{UponEvent}{EndUponEvent}{}{}[1]{\textbf{upon event} $#1$ \textbf{do}}{}
\algsetblockdefx[UponReceipt]{UponReceipt}{EndUponReceipt}{}{}[2]{\textbf{upon receipt of} $#1$ from $#2$ \textbf{do}}{}







\author{Pierre Civit}
\affiliation{
\institution{École Polytechnique Fédérale de Lausanne (EPFL)} 
\country{Switzerland}
}

\author{Seth Gilbert}
\affiliation{
\institution{NUS Singapore}
\country{Singapore}
}

\author{Rachid Guerraoui}
\affiliation{
\institution{École Polytechnique Fédérale de Lausanne (EPFL)}
\country{Switzerland}
}

\author{Jovan Komatovic}
\affiliation{
\institution{École Polytechnique Fédérale de Lausanne (EPFL)}
\country{Switzerland}
}

\author{Anton Paramonov}
\affiliation{
\institution{École Polytechnique Fédérale de Lausanne (EPFL)}
\country{Switzerland}
}

\author{Manuel Vidigueira}
\affiliation{
\institution{École Polytechnique Fédérale de Lausanne (EPFL)}
\country{Switzerland}
}

\title{All Byzantine Agreement Problems are Expensive}

\maketitle

\input{sections/abstract}
\input{sections/introduction}
\input{sections/preliminaries}

\input{sections/anton_proof}
\input{sections/weak_agreement}
\input{sections/solvability}
\input{sections/related}

\input{sections/conclusion}

\bibliographystyle{acm}
\bibliography{references}

\appendix
\input{sections/formal_proof_lower_bound}

\input{appendix/canonical_containment}
\input{appendix/reduction_proof}

\end{document}

%% file: sty/packages.tex
\usepackage{libertine}

\usepackage{datetime}
\usepackage{url}
\usepackage{paralist}

\usepackage{tabularx}

\usepackage[font=small,labelfont=bf,justification=centering]{caption}

\usepackage{stmaryrd}

\usepackage{amsmath,amsthm}
\newtheorem{theorem}{Theorem}

\theoremstyle{definition}
\newtheorem{lemma}{Lemma}

\theoremstyle{definition}
\newtheorem{definition}{Definition}

\theoremstyle{definition}

\theoremstyle{definition}
\newtheorem{corollary}{Corollary}

\newenvironment{sketch}{\textsc{Proof Sketch.}}{\hfill$\square$\smallskip}

\theoremstyle{definition}

\theoremstyle{definition}

\usepackage{cleveref}
\crefname{section}{\S}{\S\S}
\Crefname{section}{\S}{\S\S}
\crefformat{section}{\S#2#1#3}
\Crefname{line}{Line}{line}
\crefname{line}{Line}{line}
\Crefname{assumption}{Assumption}{assumption}

\crefname{table}{Table}{Tables}

\usepackage{graphicx}
\usepackage{subfig}
\usepackage{arydshln}

\usepackage{multirow}

\usepackage{threeparttable}

\usepackage{algorithmicx}
\usepackage{algorithm} 
\usepackage{algpseudocode}
\usepackage{algcompatible}

%% file: sty/defs.tex
\usepackage{xspace}

\newcommand{\pstatezero}[1]{0_{#1}}
\newcommand{\pstateone}[1]{1_{#1}}

\newcommand{\fstate}[2]{\mathsf{state}(#1, #2)}
\newcommand{\fmsgs}[2]{\mathsf{sent}(#1, #2)}
\newcommand{\fmsgso}[2]{\mathsf{send\_omitted}(#1, #2)}
\newcommand{\fmsgr}[2]{\mathsf{received}(#1, #2)}
\newcommand{\fmsgro}[2]{\mathsf{receive\_omitted}(#1, #2)}

\newcommand{\fswitcho}[2]{{#1}^{#2 \rightleftarrows}}

\usepackage{xcolor}
\usepackage{framed}
\definecolor{lightgray}{gray}{0.90}
\renewenvironment{leftbar}[1][\hsize]
{%
\MakeFramed{\hsize#1\advance\hsize-\width\FrameRestore}%
}
{\endMakeFramed}

\algsetblockdefx[UponSend]{UponSend}{EndUponSend}{}{}[1]{\textbf{upon} \textsc{send}($#1$)}{\phantom}
\algtext*{EndUponSend}

\algsetblockdefx[UponDeliver]{UponDeliver}{EndUponDeliver}{}{}[1]{\textbf{upon} \textsc{SecureDeliver}($#1$)}{\phantom}
\algtext*{EndUponDeliver}

\algsetblockdefx[UponInit]{UponInit}{EndUponInit}{}{}[0]{\textbf{upon} \textsc{init}}{\phantom}
\algtext*{EndUponInit}

\algsetblockdefx[UponExists]{UponExists}{EndUponExists}{}{}[2]{\textbf{upon exists} $#1$ such that $#2$}{\phantom}
\algtext*{EndUponExists}

\algsetblockdefx[Procedure]{Procedure}{EndProcedure}{}{}[2]{\textbf{procedure} {#1}($#2$)}{\phantom}
\algtext*{EndProcedure}
\algsetblockdefx[Function]{Function}{EndFunction}{}{}[2]{\textbf{function} {#1}($#2$)}{\phantom}
\algtext*{EndFunction}

\algnewcommand{\BlueComment}[1]{\textcolor{blue}{\hfill\(\triangleright\) #1}}
\algnewcommand{\LineComment}[1]{\State \(\triangleright\) #1}


%% file: sty/code.tex
\usepackage{listings}
\crefname{lstlisting}{listing}{listings}
\Crefname{lstlisting}{Listing}{Listings}

\crefname{code}{line}{lines}
\Crefname{code}{Line}{Lines}

\usepackage{xcolor}

\definecolor{mygreen}{rgb}{0.254,0.572,0.294}
\definecolor{mygray}{rgb}{0.5,0.5,0.5}
\definecolor{myorange}{rgb}{1,0.35,0}
\definecolor{mymauve}{rgb}{0.58,0,0.82}
\definecolor{myblue}{rgb}{0.2,0.4,0.6}

\definecolor{rakos4orange}{RGB}{255,165,0}
\definecolor{rakos4blue}{RGB}{14,48,173}
\definecolor{rakos4lblue}{RGB}{92,172,238}
\definecolor{rakos4dgray}{RGB}{77,77,77}
\definecolor{plainred}{RGB}{211,63,63}
\definecolor{plainorange}{RGB}{221,105,41}

\usepackage{bold-extra} 

\lstdefinelanguage{Golang}%
  {morekeywords=[1]{package,import,struct,defer,panic,%
     recover,select,var,const,iota,},%
   morekeywords=[2]{string,uint,uint8,uint16,uint32,uint64,int,int8,int16,%
     int32,int64,bool,float32,float64,complex64,complex128,byte,rune,uintptr,%
     error,interface,message,node},%
   morekeywords=[3]{map,slice,make,new,nil,len,cap,copy,close,true,false,%
     delete,append,real,imag,complex,chan,},%
   morekeywords=[4]{break,continue,goto,switch,case,fallthrough,%
    default,},%
   morekeywords=[5]{Println,Printf,Error,Send},%
   sensitive=true,%
   morecomment=[l]{//},%
   morecomment=[s]{/*}{*/},%
   morestring=[b]",%
   morestring=[s]{`}{`},%
   }

%% file: sections/abstract.tex
Byzantine agreement, arguably the most fundamental problem in distributed computing, operates among $n$ processes, out of which $t < n$ can exhibit arbitrary failures. 
The problem states that all correct (non-faulty) processes must eventually decide (termination) the same value (agreement) from a set of admissible values defined by the proposals of the processes (validity).
Depending on the exact version of the validity property, Byzantine agreement comes in different forms, from Byzantine broadcast to strong and weak consensus, to modern variants of the problem introduced in today's blockchain systems.
Regardless of the specific flavor of the agreement problem, its communication cost is a fundamental metric whose improvement has been the focus of decades of research. 
The Dolev-Reischuk bound, one of the most celebrated results in distributed computing, proved 40 years ago that, at least for Byzantine broadcast, no deterministic solution can do better than $\Omega(t^2)$ exchanged messages in the worst case.
Since then, it remained unknown whether the quadratic lower bound extends to seemingly weaker variants of Byzantine agreement.
This paper answers the question in the affirmative, closing this long-standing open problem.
Namely, we prove that \emph{any} non-trivial agreement problem requires $\Omega(t^2)$ messages to be exchanged in the worst case.
To prove the general lower bound, we determine the weakest Byzantine agreement problem and show, via a novel indistinguishability argument, that it incurs $\Omega(t^2)$ exchanged messages.

%% file: sections/introduction.tex
\section{Introduction}

Byzantine agreement~\cite{LSP82} is a foundational problem of distributed computing.
Its importance stems from the fact that Byzantine agreement lies at the heart of state machine replication~\cite{CL02,adya2002farsite,abd2005fault,kotla2004high,veronese2011efficient,amir2006scaling,kotla2007zyzzyva,malkhi2019flexible,momose2021multi}, distributed key generation~\cite{AbrahamJMMST21,ShresthaBKN21,Kokoris-KogiasM20,DasYXMK022}, secure multi-party computation~\cite{DBLP:conf/tcc/DeligiosHL21,DBLP:conf/eurocrypt/FitziGMR02,DBLP:conf/crypto/GennaroIKR02}, as well as various distributed services~\cite{galil1987cryptographic, gilbert2010rambo}.
Recent years have witnessed a renewed interest in Byzantine agreement due to the emergence of blockchain systems~\cite{abraham2016solida,chen2016algorand,abraham2016solidus,luu2015scp,correia2019byzantine,CGL18,buchman2016tendermint}.
Formally, the agreement problem is defined in a distributed system of $n$ processes; up to $t < n$ processes can be \emph{faulty}, whereas the rest are \emph{correct}.
Correct processes behave according to the prescribed deterministic protocol; faulty processes can deviate arbitrarily from it.
Byzantine agreement exposes the following interface:
\begin{compactitem}
    \item \textbf{input} $\mathsf{propose}(v \in \mathcal{V}_I)$: a process proposes a value $v$ from a (potentially infinite) set $\mathcal{V}_I$.

    \item \textbf{output} $\mathsf{decide}(v' \in \mathcal{V}_O)$: a process decides a value $v'$ from a (potentially infinite) set $\mathcal{V}_O$.
\end{compactitem}
Byzantine agreement ensures the following properties:
\begin{compactitem}
    \item \emph{Termination:} Every correct process eventually decides.


    \item \emph{Agreement:} No two correct processes decide different values.
\end{compactitem}
To preclude a trivial solution in which processes agree on a predetermined value, Byzantine agreement requires an additional property -- \emph{validity} -- that specifies which decisions are admissible.

The exact definition of the validity property yields a specific agreement problem. 
For example, Byzantine broadcast~\cite{Wan2020,Wan2023a,abraham2021good,Nayak2020a} ensures \emph{Sender Validity}, i.e., if the predetermined sender is correct, then its proposed value must be decided by a correct process.
Weak consensus~\cite{yin2019hotstuff,lewis2022quadratic,civit2022byzantine,BKM19} guarantees only \emph{Weak Validity}, i.e., if all processes are correct and they all propose the same value, that value is the sole admissible decision.
Other notable Byzantine agreement problems include (1) strong consensus~\cite{LSP82,civit2022byzantine,CGL18}, ensuring that, if all correct processes propose the same value, that value must be decided, (2) interactive consistency~\cite{LSP82,fischer1981lower,ben2003resilient}, where correct processes agree on the proposals of all $n$ processes, and (3) agreement problems employed in today's blockchain systems~\cite{Cachin2001,BKM19,yin2019hotstuff}, which require the decided value to satisfy a globally verifiable condition (e.g., the value is a  transaction correctly signed by the issuing client).

\paragraph{The worst-case communication cost of Byzantine agreement}
Motivated by practical implications, one of the most studied aspects of Byzantine agreement is its communication cost.
Since the inception of Byzantine agreement, research has been focused on minimizing the number of exchanged bits of information~\cite{dolev1985bounds,validity_podc,lewis2022quadratic,wan2023amortized,civit2022byzantine,everyBitCounts,DBLP:journals/iandc/CoanW92,berman1992bit,Chen2021,Nayak2020a,Abraham2023a}.
However, there are intrinsic limits.
The seminal Dolev-Reischuk bound~\cite{dolev1985bounds} proves that Byzantine broadcast cannot be solved unless $\Omega(t^2)$ messages are exchanged in the worst case.
(This naturally applies to any problem to which Byzantine broadcast can be reduced with $o(t^2)$ messages.)
The result of~\cite{dolev1985bounds} is shown for  any Byzantine broadcast algorithm that operates in \emph{synchrony}, where the message delays are known.
Inherently, this lower bound applies to weaker network models as well. 
Concretely, it extends to \emph{partial synchrony}~\cite{DLS88}, in which the communication is asynchronous (with arbitrary message delays) until some unknown point in time, after which it becomes synchronous.
(Byzantine agreement is known to be unsolvable in full asynchrony~\cite{fischer1985impossibility}.)

While the Dolev-Reishcuk bound answers the question of what the necessary message cost is for Byzantine broadcast, it is not general, i.e., it does not hold for \emph{any} specific non-trivial agreement problem.
(An agreement problem is trivial if there exists an always-admissible value that can be decided immediately, i.e., without any communication.)
For instance, the Dolev-Reischuk bound does not apply to weak consensus.
Thus, whether all non-trivial agreement problems require a quadratic number of messages remains unknown.
In this paper, we answer this long-standing question in the affirmative.

\begin{theorem} \label{theorem:general_lower_bound_main}
No (non-trivial) Byzantine agreement problem can be solved with fewer than $\Omega(t^2)$ exchanged messages in the worst case even in synchrony.
\end{theorem}


To prove our general lower bound, we study binary ($\mathcal{V}_I = \mathcal{V}_O = \{0, 1\}$) weak consensus in synchrony.
Namely, we first prove an $\Omega(t^2)$ lower bound on the number of exchanged messages for weak consensus.
Then, to generalize the bound, we prove that weak consensus is the weakest agreement problem by presenting a reduction from it to any (solvable and non-trivial) agreement problem.
As a byproduct, the reduction allows us to define the entire landscape of solvable (and unsolvable) agreement problems, thus unifying all previous results on the solvability of Byzantine agreement. 
(We believe this result to be important in its own right.)

\paragraph{The fundamental challenge of weak consensus}
Recall that the \emph{Weak Validity} property of weak consensus guarantees only that, if all processes are correct and they all propose the same value, that value must be decided.
This is a very weak requirement: picking $1$ as the decision is always allowed except in a \emph{single} execution $\mathcal{E}$ where all processes are correct and they all propose $0$.
Hence, any weak consensus algorithm needs only to distinguish \emph{two} scenarios: either the execution is (1) $\mathcal{E}$, deciding $0$, or (2) non-$\mathcal{E}$, deciding $1$.
This observation was the starting point for our conjecture that weak consensus is the \emph{weakest} (non-trivial) agreement problem (which we prove in this paper), implying that any lower bound for weak consensus also applies to all other agreement problems.

To illustrate the difficulty of proving a quadratic lower bound for weak consensus, we briefly discuss the common point in the classical proof techniques exploited for similar results (namely,~\cite{dolev1985bounds} and~\cite{validity_podc}) and explain why those techniques cannot be easily adapted to weak consensus in synchrony.
The crux of those proof techniques consists in showing that, unless $\Omega(t^2)$ messages are exchanged, there necessarily exists an execution $\mathcal{E}_1$ in which some correct process $p$ decides $1$ without receiving any message.
The second step of the proof consists of constructing another execution $\mathcal{E}_0$ in which (1) $p$ is correct and receives no messages, and (2) some correct process $q \neq p$ decides $0$.
As $p$ cannot distinguish $\mathcal{E}_0$ from $\mathcal{E}_1$, $p$ decides $1$ in $\mathcal{E}_0$, thus violating \emph{Agreement}.
Unfortunately, while elegant, this approach cannot be directly adapted to weak consensus in synchrony as both $\mathcal{E}_0$ and $\mathcal{E}_1$ inevitably contain detectable faults.
Therefore, nothing prevents a weak consensus algorithm from deciding $1$ in \emph{both} $\mathcal{E}_0$ and $\mathcal{E}_1$, making the aforementioned reasoning inapplicable.
Intuitively, the main difficulty in proving a quadratic lower bound for weak consensus is that \emph{any} detectable misbehavior immediately allows an algorithm to choose a predetermined ``default'' value.


\paragraph{Technical overview.}
To prove an $\Omega(t^2)$ lower bound for weak consensus in the Byzantine failure model, we show that the bound holds even with only \emph{omission} failures.
An omission-faulty process can only misbehave by failing to receive or send some messages, but not by behaving maliciously.
(In contrast to Byzantine processes, it is reasonable to make claims about the behavior of omission-faulty processes as they are still \emph{honest}, i.e., they never act malevolently.)
Our proof utilizes in a novel way the standard concept of \emph{isolation}~\cite{dolev1985bounds,validity_podc,AbrahamStern22,Abraham2023revisited,Abraham2019c,hadzilacos1991message}, in which a small subset of omission-faulty processes starts (from some round onward) ``dropping'' all messages received from outside the set.
Concretely, we obtain our bound through a sequence of four critical observations about what happens when \emph{multiple} groups of processes are isolated.
Suppose that there are three groups: group $A$, which is fully correct and sends $o(t^2)$ messages, and groups $B$ and $C$, which are (separately) isolated from rounds $k_B$ and $k_C$, respectively.
We observe that:
\begin{compactenum}
    \item In any execution in which group $B$ (resp., $C$) is isolated, correct processes from $A$ and a majority of processes from $B$ (resp., $C$) must decide the \emph{same} bit; otherwise, we could design an execution which violates the properties of weak consensus.

    \item If both $B$ and $C$ are isolated from round $1$, group $A$ must decide some ``default'' bit \emph{independently} of their proposals, i.e., group $A$ either always decides 0 or always decides 1 whenever $B$ and $C$ are isolated from round $1$.

    \item At some round $R$ in the execution, $A$ must stop deciding the default bit \emph{even if there are faults afterward} (e.g., even if $B$ and $C$ are isolated).
    For example, if the default bit is $1$, but all processes propose $0$ and act correctly until the end, then, by an interpolation argument, all correct processes must at some round $R$ direct their strategy towards deciding $0$ (otherwise, they would violate \emph{Weak Validity}).

    \item Isolating $B$ and $C$ at the same round (e.g., $k_C = k_B = R$) or one round apart (e.g., $k_B = k_C + 1 = R$) is indistinguishable for processes in $B$ or $C$.
    Thus, we can create a situation where processes in $C$ decide the default bit $1$, while processes in $B$ choose $0$.
    In this situation, processes in $A$ necessarily violate the statement of the first observation: if they decide $1$, they disagree with $B$; if they decide $0$, they disagree with $C$.
\end{compactenum}

To generalize our lower bound, we then show that weak consensus is reducible at $0$ message cost to any solvable and non-trivial agreement problem in synchrony.
This reduction is possible because, for any Byzantine agreement problem that is non-trivial and synchronously solvable, its specific validity property must follow a certain structure.
Concretely, we define a simple combinatorial condition -- the \emph{containment condition} -- which we prove to be a necessary condition for synchronously solvable non-trivial agreement problems.
Interestingly, the containment condition is also \emph{sufficient}, enabling us to devise the general solvability theorem for Byzantine agreement problems. 


\paragraph{Roadmap.}
We state the system model and preliminaries in \Cref{section:preliminaries}.
In \Cref{section:lower_bound_weak_consensus}, we prove the $\Omega(t^2)$ lower bound on exchanged messages for weak consensus.
A generalization of the bound to all (solvable) non-trivial agreement problems is provided in \Cref{section:weak_consensus}.
In \Cref{section:solvability}, we present the general solvability theorem for Byzantine agreement problems.
We provide an overview of related work in \Cref{section:related_work}, and conclude the paper in \Cref{section:conclusion}.
The optional appendix contains omitted proofs.

%% file: sections/preliminaries.tex
\section{System Model \& Preliminaries} \label{section:preliminaries}

\paragraph{Processes \& adversary.}
We consider a static system $\Pi = \{p_1, ..., p_n\}$ of $n$ processes, where each process acts as a deterministic state machine.
Moreover, we consider a \emph{static adversary} which can corrupt up to $t < n$ processes before each run of the system.\footnote{
Note that a lower bound proven for a static adversary trivially applies to a stronger adaptive adversary which can corrupt processes during (and not only before) a run of the system.}
A corrupted process can behave arbitrarily; a non-corrupted process behaves according to its state machine.
We say that a corrupted process is \emph{faulty}, whereas a non-corrupted process is \emph{correct}.

\paragraph{Synchronous environment.}
Computation unfolds in synchronous rounds.
In each round $1, 2, ... \in \mathbb{N}$, each process (1) performs (deterministic) local computations, (2) sends (possibly different) messages to (a subset of) the other processes, and (3) receives the messages sent to it in the round.
We assume authenticated channels: the receiver of a message is aware of the sender's identity.



\paragraph{Executions.}
Each execution of any algorithm is uniquely identified by (1) the sets of correct and faulty processes, and (2) the messages faulty processes send (or do not send) in each round.
Given any algorithm $\mathcal{A}$, $\mathit{execs}(\mathcal{A})$ denotes the set of all $\mathcal{A}$'s executions with no more than $t$ faulty processes.
Lastly, $\mathit{Correct}_{\mathcal{A}}(\mathcal{E})$ denotes the set of correct processes in any execution $\mathcal{E} \in \mathit{execs}(\mathcal{A})$.

\paragraph{Message complexity.}
Let $\mathcal{A}$ be any algorithm and let $\mathcal{E}$ be any execution of $\mathcal{A}$.
The message complexity of $\mathcal{E}$ is the number of messages sent by correct processes throughout the entire execution $\mathcal{E}$.
(Note that all messages count towards the message complexity of $\mathcal{E}$, even those sent after all correct processes have already decided.)
The \emph{message complexity} of $\mathcal{A}$ is then defined as
\begin{equation*}
    \max_{\mathcal{E} \in \mathit{execs}(\mathcal{A})}\bigg\{\text{the message complexity of } \mathcal{E} \bigg\}.
\end{equation*}

%% file: sections/anton_proof.tex
\section{Lower Bound on Message Complexity of Weak Consensus} \label{section:lower_bound_weak_consensus}

To prove our general lower bound, we first show a quadratic lower bound for weak consensus:

\begin{theorem} \label{theorem:lower_bound_weak_validity}
Any weak consensus algorithm has $\Omega(t^2)$ message complexity.
\end{theorem}

In order to prove \Cref{theorem:lower_bound_weak_validity}, we show a strictly stronger lower bound for the omission failure model in which processes can only fail by ``dropping'' some messages they send or receive, but not by behaving maliciously.

\paragraph{Omission failures.}
In (only) this section, we consider \emph{omission failures}.
A static adversary corrupts up to $t < n$ processes before each execution.
A corrupted process can commit:
\begin{compactitem}
    \item \emph{send-omission} faults, by not sending some messages it is supposed to send; or

    \item \emph{receive-omission} faults, by not receiving some messages it is supposed to receive.
\end{compactitem}
Note that a faulty process cannot misbehave in an arbitrary manner, i.e., it acts according to its state machine at all times.
Moreover, corrupted processes are unaware that they are corrupted, i.e., they do not know if or when they omitted some messages.
Corrupted processes are said to be \emph{faulty}, whereas non-corrupted processes are said to be \emph{correct}.

Two executions are said to be \emph{indistinguishable} to a (correct or faulty) process if and only if (1) the process has the same proposal in both executions and (2) the process receives identical messages in each round of both executions.
Note that, given two executions indistinguishable to some process, the process's actions in each round of both executions are \emph{identical} due to the process's determinism.
Concretely, if two $k$-round-long ($k \in \mathbb{N}$) executions are indistinguishable to a process $p_i$, then (1) $p_i$'s internal states at the start of the $(k + 1)$-st round of both executions are identical, and (2) the sets of all messages sent (including those that are omitted) in the $(k + 1)$-st round of both executions are identical.
We relegate a precise definition of the omission failure model to \Cref{section:lower_bound_formal}.

\paragraph{Notation \& remarks.}
Given any set of processes $G$, let $\bar{G} = \Pi \setminus{G}$.
If a faulty process omits sending (resp., omits receiving) some message $m$, we say that the process \emph{send-omits} (resp., \emph{receive-omits}) $m$.
Note that, in the omission failure model, it is reasonable to make claims about the behaviors of faulty processes as they always behave according to their state machine.
Finally, observe that any weak consensus algorithm provides guarantees \emph{only} to correct processes, i.e., it is possible for faulty processes to not terminate or to disagree (among themselves or with correct processes).

\paragraph{Proof of \Cref{theorem:lower_bound_weak_validity}.}

As previously mentioned, we prove a quadratic lower bound for weak consensus by showing that the problem requires at least $\frac{t^2}{32}$ messages even with omission failures:

\begin{restatable}{lemma}{lowerboundomission}\label{lemma:lower_bound_omission}
Any omission-resilient weak consensus algorithm has at least $\frac{t^2}{32}$ message complexity.
\end{restatable}

We prove \Cref{lemma:lower_bound_omission} by contradiction.
Fix any $n$ and $t$ such that $t \in [8, n - 1]$.
(Without loss of generality, we consider $t$ divisible by $8$.)
Fix any weak consensus algorithm $\mathcal{A}$ which (1) tolerates $t$ omission failures and works among $n$ processes, and (2) whose message complexity is less than $\frac{t^2}{32}$.
This implies that correct processes send fewer than $\frac{t^2}{32}$ messages in \emph{every} execution of $\mathcal{A}$.
\Cref{table:notation_weak_consensus} introduces notation we rely on throughout the proof.

\begin{table} [bh]
    \begin{tabular}{|l|m{8.5cm}|}
    \multicolumn{1}{l}{\emph{Notation}} & \multicolumn{1}{l}{\emph{Definition}} \\ 
    \hline
    \hline
        $(A, B, C)$ & Any partition of $\Pi$ such that (1) $|B| = \frac{t}{4}$, and (2) $|C| = \frac{t}{4}$ (naturally, $|A| = n - \frac{t}{2}$).
         \\ \hline
        $\mathcal{E}_0$ & The infinite execution of $\mathcal{A}$ in which (1) all processes propose $0$, and (2) all processes are correct. \\ \hline
        $\mathcal{E}_0^{B(k)}, k \in \mathbb{N}$ & The infinite execution of $\mathcal{A}$ in which (1) all processes propose $0$, (2) processes from $A \cup C$ are correct, and (3) group $B$ is isolated from round $k$. \\ \hline
        $\mathcal{E}_0^{C(k)}, k \in \mathbb{N}$ & The infinite execution of $\mathcal{A}$ in which (1) all processes propose $0$, (2) processes from $A \cup B$ are correct, and (3) group $C$ is isolated from round $k$. \\ \hline
        $\mathcal{E}_1^{C(1)}$ & The infinite execution of $\mathcal{A}$ in which (1) all processes propose $1$, (2) processes from  $A \cup B$ are correct, and (3) group $C$ is isolated from round $1$.  \\ \hline
    \end{tabular}
    \caption{Notation table for the lower bound for weak consensus.\\ (The concept of group isolation is described in \Cref{definition:isolation}.)}
    \label{table:notation_weak_consensus}
\end{table}

First, let us introduce the concept of \emph{isolation}, which we use extensively throughout the proof.

\begin{definition} [Isolation] \label{definition:isolation}
A group $G \subsetneq \Pi$ of $|G| \leq t$ processes is \emph{isolated from some round $k \in \mathbb{N}$} in an execution $\mathcal{E}$ of $\mathcal{A}$ if and only if, for every process $p_G \in G$, the following holds:
\begin{compactitem}
    \item $p_G$ is faulty in $\mathcal{E}$; and
    
    \item $p_G$ does not send-omit any message in $\mathcal{E}$; and

    \item for every message $m$ sent by any process $p_m$ to $p_G$ in any round $k' \in \mathbb{N}$ of $\mathcal{E}$, $p_G$ receive-omits $m$ in $\mathcal{E}$ if and only if (1) $p_m \in \bar{G}$, and (2) $k' \geq k$.
\end{compactitem}
\end{definition}

Intuitively, a group $G$ is isolated from some round $k$ if and only if no process $p_G \in G$ receives any message from outside of $G$ in any round $k' \geq k$, i.e., $p_G$ only receives messages sent by processes in $G$ from round $k$ onward; other than these receive-omission faults, $p_G$ commits no other faults.
\Cref{figure:helper_isolation_illustration} illustrates the concept of isolation.

\begin{figure}[ht]
    \centering
    \begin{minipage}{1\textwidth}
        \centering
        \includegraphics[width=0.9\linewidth]{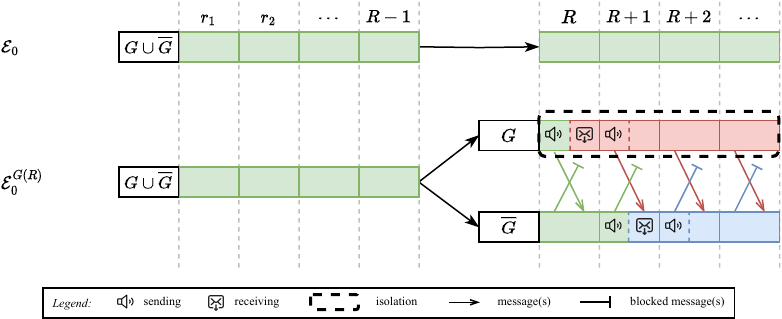}
        \caption{Illustration of \Cref{definition:isolation}.
        The colors represent the local behaviors of processes.
        Execution $\mathcal{E}_0$ has no faults.
        Execution $\mathcal{E}_0^{G(R)}$ proceeds identically to $\mathcal{E}_0$, sending the same messages (green color) up until round $R$ (inclusive).
        However, group $G$ is \emph{isolated} at round $R$, causing it to drop all messages from group $\overline{G}$ from then on.
        This (potentially) changes $G$'s sending behavior from round $R+1$ onward (red color).
        By propagation, group $\overline{G}$ is then (potentially) affected by $G$'s new sending behavior (red color), causing $\overline{G}$ to deviate from $\mathcal{E}_0$ in the messages it sends from round $R+2$ onward (blue color).
        }
        \label{figure:helper_isolation_illustration}
    \end{minipage}
\end{figure}

Let $(X, Y, Z)$ be any partition of $\Pi$ such that $|Y| = \frac{t}{4}$ and $|Z| \leq \frac{t}{4}$.
The following lemma proves that in any infinite execution $\mathcal{E}$ of $\mathcal{A}$ in which processes from $X$ are correct and processes from $Y \cup Z$ are faulty, more than half of processes from $Y$ decide the same bit as (all) processes from $X$.
If this was not the case, we could construct an execution that demonstrates that $\mathcal{A}$ is not a correct weak consensus algorithm.
We formally prove the lemma in \Cref{section:lower_bound_formal}.

\begin{restatable}{lemma}{genericisolation}\label{lemma:lower_bound_helper}
Let $(X, Y, Z)$ be any partition of $\Pi$ such that (1) $|Y| = \frac{t}{4}$, and (2) $|Z| \leq \frac{t}{4}$ (naturally, $|X| = n - |Y| - |Z|$).
Moreover, let $\mathcal{E}$ be any infinite execution of $\mathcal{A}$ such that:
\begin{compactitem}
    \item processes from $X$ are correct in $\mathcal{E}$, whereas processes from $Y \cup Z$ are faulty in $\mathcal{E}$; and

    \item all processes from $X$ decide the same bit $b_X$ (to satisfy \emph{Termination} and \emph{Agreement}); and 

    \item group $Y$ is isolated from some round $k \in \mathbb{N}$ in $\mathcal{E}$.
\end{compactitem}
Then, there exists a set $Y' \subseteq Y$ of $|Y'| > \frac{|Y|}{2}$ processes such that all processes in $Y'$ decide $b_X$ in $\mathcal{E}$.
\end{restatable}
\begin{sketch}
For every process $p \in Y$, let $\mathcal{M}_{X \to p}$ denote the set of all messages which are (1) sent by any process $p' \in X$ in $\mathcal{E}$, and (2) receive-omitted by $p$ in $\mathcal{E}$; as $p \in Y$ and group $Y$ is isolated from round $k$ in $\mathcal{E}$, every message $m \in \mathcal{M}_{X \to p}$ is sent in some round $k' \geq k$.
For every set $Y'' \subseteq Y$, let $\mathcal{M}_{X \to Y''} = \bigcup\limits_{p \in Y''} \mathcal{M}_{X \to p}$.
As correct processes (i.e., processes from group $X$) send fewer than $\frac{t^2}{32}$ messages in $\mathcal{E}$, $|\mathcal{M}_{X \to Y}| < \frac{t^2}{32}$.
Therefore, there does not exist a set $Y^* \subseteq Y$ of $|Y^*| \geq \frac{|Y|}{2}$ processes such that, for every process $p_{Y^*} \in Y^*$, $|\mathcal{M}_{X \to p_{Y^*}}| \geq \frac{t}{2}$.
This implies that there exists a set $Y' \subseteq Y$ of $|Y'| > \frac{|Y|}{2}$ processes such that, for every process $p_{Y'} \in Y'$, $|\mathcal{M}_{X \to p_{Y'}}| < \frac{t}{2}$.

Fix any process $p_{Y'} \in Y'$.
By contradiction, suppose that $p_{Y'}$ does not decide $b_X$ in $\mathcal{E}$.
Let $\mathcal{S}$ denote the set of all processes whose messages $p_{Y'}$ receive-omits in (any round $k' \geq k$ of) $\mathcal{E}$; note that $|\mathcal{S} \cap X| < \frac{t}{2}$ (since $|\mathcal{M}_{X \to p_{Y'}}| < \frac{t}{2}$) and $\mathcal{S} \subsetneq X \cup Z$.
Let us construct another infinite execution $\mathcal{E}'$ of $\mathcal{A}$ following the (sequentially-executed) steps below:
\begin{compactenum}
    \item Processes in $\mathcal{S} \cup Y \cup Z \setminus{\{ p_{Y'} \}}$ are faulty in $\mathcal{E}'$, whereas all other processes are correct.

    \item Then, we set $\mathcal{E}' \gets \mathcal{E}$: every process (at first) behaves in the same manner as in $\mathcal{E}$.

    \item For every message $m$ such that $p_{Y'}$ receive-omits $m$ in $\mathcal{E}$, $m$ is send-omitted in $\mathcal{E}'$.
    That is, the sender of $m$ is responsible for $p_{Y'}$ not receiving $m$ in $\mathcal{E}'$.
\end{compactenum}
Observe that $p_{Y'}$ is indeed correct in $\mathcal{E}'$ as (1) $p_{Y'}$ does not commit any send-omission faults (since $p_{Y'}$ does not commit those faults in $\mathcal{E}$), and (2) $p_{Y'}$ does not commit any receive-omission faults (since every message which is receive-omitted in $\mathcal{E}$ is send-omitted in $\mathcal{E}'$).
Moreover, there are $|\mathcal{S} \cup Y \cup Z \setminus{\{p_{Y'}\}}| = |(\mathcal{S} \cap X) \cup Y \cup Z \setminus{\{p_{Y'}\}}| < \frac{t}{2} + \frac{t}{4} + \frac{t}{4} - 1 < t$ faulty processes in $\mathcal{E}'$.
Furthermore, there exists a process $p_X \in X$ which is correct in $\mathcal{E}'$ as $|\mathcal{S} \cap X| < \frac{t}{2}$, $|X| \geq n - \frac{t}{2}$ and $n > t$.
Finally, neither $p_{Y'}$ nor $p_X$ can distinguish $\mathcal{E}'$ from $\mathcal{E}$ as their behaviors in $\mathcal{E}'$ and $\mathcal{E}$ are identical.\footnote{Recall that process $p_{Y'}$ is unaware of receive-omission failures it commits in $\mathcal{E}$. Therefore, the fact that $p_{Y'}$ does not commit receive-omission failures in $\mathcal{E}'$ does not allow $p_{Y'}$ to distinguish $\mathcal{E}'$ from $\mathcal{E}$.}
Therefore, either \emph{Termination} (if $p_{Y'}$ does not decide) or \emph{Agreement} (if $p_{Y'}$ decides $1 - b_X$) is violated in $\mathcal{E}'$, which contradicts the fact that $\mathcal{A}$ is a correct weak consensus algorithm.
\end{sketch}

Next, we define \emph{mergeable} executions.

\begin{definition} [Mergeable executions] \label{definition:mergeable}
Any two infinite executions $\mathcal{E}_{0}^{B(k_1)}$ ($k_1 \in \mathbb{N}$) and $\mathcal{E}_{b}^{C(k_2)}$ ($b \in \{0, 1\}$, $k_2 \in \mathbb{N}$) are \emph{mergeable} if and only if:
\begin{compactitem}
    \item $k_1 = k_2 = 1$; or

    \item $|k_1 - k_2| \leq 1$ and $b = 0$.
\end{compactitem}
\end{definition}

In brief, executions $\mathcal{E}_0^{B(k_1)}$ and $\mathcal{E}_{b}^{C(k_2)}$ (which are defined in \Cref{table:notation_weak_consensus}) are mergeable if (1) group $B$ (resp., $C$) is isolated from round $1$ in $\mathcal{E}_0^{B(k_1)}$ (resp., $\mathcal{E}_{b}^{C(k_2)}$), or (2) $b = 0$ and groups $B$ and $C$ are isolated at most one round apart in their respective executions.
Note that all processes from group $A$ are correct in any two mergeable executions.
The following lemma proves that processes from group $A$ decide identically in any two mergeable executions, and it represents a crucial intermediate result in proving our lower bound.
We formally prove the lemma in \Cref{section:lower_bound_formal}.
An illustration of its application can be seen in \Cref{figure:helper_final_lemma_illustration}.

\begin{restatable}{lemma}{either}
\label{lemma:either_or}
Let $\mathcal{E}_0^{B(k_1)}$ ($k_1 \in \mathbb{N}$) and $\mathcal{E}_{b}^{C(k_2)}$ ($b \in \{0, 1\}, k_2 \in \mathbb{N}$) be any two mergeable executions.
Let any process from group $A$ decide $b_1$ (resp., $b_2$) in $\mathcal{E}_0^{B(k_1)}$ (resp., $\mathcal{E}_{b}^{C(k_2)}$).
Then, $b_1 = b_2$.
\end{restatable}
\begin{sketch}
For $\mathcal{A}$ to satisfy \emph{Termination} and \emph{Agreement}, all processes from group $A$ decide $b_1$ (resp., $b_2$) in $\mathcal{E}_0^{B(k_1)}$ (resp., $\mathcal{E}_b^{C(k_2)}$).
Given the partition $(A \cup C, B, \emptyset)$ of $\Pi$ and the execution $\mathcal{E}_0^{B(k_1)}$, \Cref{lemma:lower_bound_helper} proves that there exists a set $B' \subseteq B$ of more than $\frac{|B|}{2}$ processes such that every process $p_{B'} \in B'$ decides $b_1$ in $\mathcal{E}_0^{B(k_1)}$.
Similarly, given the partition $(A \cup B, C, \emptyset)$ of $\Pi$ and the execution $\mathcal{E}_b^{C(k_2)}$, \Cref{lemma:lower_bound_helper} proves that there exists a set $C' \subseteq C$ of more than $\frac{|C|}{2}$ processes such that every process $p_{C'} \in C'$ decides $b_2$ in $\mathcal{E}_b^{C(k_2)}$.

We now construct another infinite execution $\mathcal{E}$ of $\mathcal{A}$:
\begin{compactenum}
    \item Processes from group $A$ are correct, whereas processes from $B \cup C$ are faulty.
    
    \item All processes from $A \cup B$ propose $0$, whereas all processes from group $C$ propose $b$.

    \item Every process $p_B \in B$ (resp., $p_C \in C$) behaves in the same manner as in $\mathcal{E}^{B(k_1)}_0$ (resp., $\mathcal{E}^{C(k_2)}_b$).
    Let us elaborate on why this step of the construction is valid:
    \begin{compactitem}
        \item Suppose that $k_1 = k_2 = 1$.
        Due to the construction of $\mathcal{E}$, every process $p_B \in B$ (resp., $p_C \in C$) receives messages only from other processes in the same group $B$ (resp., $C$) in $\mathcal{E}$.
        As (1) all messages received by $p_B \in B$ (resp., $p_C \in C$) in $\mathcal{E}$ are sent in $\mathcal{E}_0^{B(1)}$ (resp., $\mathcal{E}_b^{C(1)}$), and (2) for every process $p_B' \in B$ (resp., $p_C' \in C$), the set of messages sent by $p_B'$ (resp., $p_C'$) in $\mathcal{E}$ is identical to the set of messages sent by $p_B'$ (resp., $p_C'$) in $\mathcal{E}_0^{B(1)}$ (resp., $\mathcal{E}_b^{C(1)}$), the construction step is indeed valid in this case.

        \item Suppose that $|k_1 - k_2| \leq 1$ and $b = 0$.
        As the behavior of each process from group $B$ (resp., $C$) in $\mathcal{E}$ is identical to its behavior in $\mathcal{E}_0^{B(k_1)}$ (resp., $\mathcal{E}_0^{C(k_2)}$), the set of messages received by any process $p_B \in B$ (resp., $p_C \in C$) in $\mathcal{E}$ is identical to the set of messages received by $p_B \in B$ (resp., $p_C \in C$) in $\mathcal{E}_0^{B(k_1)}$ (resp., $\mathcal{E}_0^{C(k_2)}$).
        To prove the validity of the construction step in this scenario, we show that, for each message received by any process $p_B \in B$ (resp., $p_C \in C$) in $\mathcal{E}$, that message is sent in $\mathcal{E}$.

        Without loss of generality, we fix any message $m$ received by any process $p_B \in B$ in $\mathcal{E}$.
        We denote the sender of $m$ by $p_m$.
        Note that $m$ is sent by $p_m$ in $\mathcal{E}_0^{B(k_1)}$ as $m$ is received in $\mathcal{E}_0^{B(k_1)}$.
        If $m$ is received before round $R = \min(k_1, k_2)$, $m$ is sent in $\mathcal{E}$ as, for any process $p \in \Pi$, $p$'s behaviour until (and excluding) round $R$ is identical in $\mathcal{E}$ and $\mathcal{E}_0^{B(k_1)}$.
        If $m$ is received in or after round $R$, we distinguish two possibilities:
        \begin{compactitem}
            \item Let $m$ be received before round $k_1$.
            (This is possible only if $k_1 > k_2$.)
            Hence, $m$ is received in round $R$.
            In this case, $m$ is sent in $\mathcal{E}$ as the set of messages $p_m$ sends in $\mathcal{E}$ is identical to the set of messages $p_m$ sends in $\mathcal{E}_0^{B(k_1)}$ (since the internal state of process $p_m$ at the beginning of round $R$ is identical in $\mathcal{E}$ and $\mathcal{E}_0^{B(k_1)}$).

            \item Let $m$ be received in or after round $k_1$.
            In this case, $p_m \in B$ (as group $B$ is isolated from round $k_1$ in $\mathcal{E}_0^{B(k_1)}$).
            Therefore, $m$ is sent in $\mathcal{E}$ as the behavior of every process from group $B$ in $\mathcal{E}$ is identical to its behavior in $\mathcal{E}_0^{B(k_1)}$.
        \end{compactitem}
        Note that this step of construction ensures that group $B$ (resp., $C$) is isolated from round $k_1$ (resp., $k_2$) in $\mathcal{E}$.
    \end{compactitem}
\end{compactenum}
As no process $p_{B'} \in B'$ (resp., $p_{C'} \in C'$) distinguishes $\mathcal{E}$ from $\mathcal{E}_0^{B(k_1)}$ (resp., $\mathcal{E}_b^{C(k_2)}$), all processes from $B'$ (resp., $C'$) decide $b_1$ (resp., $b_2$) in $\mathcal{E}$.
Let $b_A$ be the decision of processes from group $A$ in $\mathcal{E}$; such a decision must exist as $\mathcal{A}$ satisfies \emph{Termination} and \emph{Agreement}.
Given the partition $(A, B, C)$ of $\Pi$ and the newly constructed execution $\mathcal{E}$, \Cref{lemma:lower_bound_helper} proves that $b_1 = b_A$.
Similarly, given the partition $(A, C, B)$ of $\Pi$ and the execution $\mathcal{E}$, \Cref{lemma:lower_bound_helper} shows that $b_2 = b_A$.
As $b_1 = b_A$ and $b_A = b_2$, $b_1 = b_2$, which concludes the proof.
\end{sketch}

\Cref{lemma:either_or} implies that all processes from group $A$ decide identical values in executions $\mathcal{E}_0^{B(1)}$ and $\mathcal{E}_1^{C(1)}$ as these two executions are mergeable (see \Cref{definition:mergeable}).
Without loss of generality, the rest of the proof assumes that all processes from group $A$ decide $1$ in $\mathcal{E}_0^{B(1)}$ (and $\mathcal{E}_1^{C(1)}$).
Intuitively, the value $1$ acts as the ``default'' value for processes in $A$ if they detect faults early.
In the following lemma, we prove that there exists a round $R \in \mathbb{N}$ such that processes from group $A$ decide $1$ in $\mathcal{E}_0^{B(R)}$ and $0$ in $\mathcal{E}_0^{B(R + 1)}$.
This expresses the idea that $A$ must, at some critical round (i.e., $R+1$), abandon its initial strategy of always deciding the ``default'' value.

\begin{restatable}{lemma}{criticalround} \label{lemma:R}
There exists a round $R \in \mathbb{N}$ such that (1) all processes from group $A$ decide $1$ in $\mathcal{E}_0^{B(R)}$, and (2) all processes from group $A$ decide $0$ in $\mathcal{E}_0^{B(R + 1)}$.
\end{restatable}
\begin{proof}
Let $R_{\mathit{max}} \in \mathbb{N}$ denote the round before which all processes decide $0$ in $\mathcal{E}_0$, which is the fully correct execution with all processes proposing $0$ (see \Cref{table:notation_weak_consensus}); such a round must exist for $\mathcal{A}$ to satisfy \emph{Termination} and \emph{Weak Validity}.
Hence, all processes from group $A$ decide $0$ in $\mathcal{E}_0^{B(R_{\mathit{max}})}$.
By our assumption, all processes from group $A$ decide $1$ in $\mathcal{E}_0^{B(1)}$.
Therefore, there exists a round $R \in [1, R_{\mathit{max}})$ which satisfies the statement of the lemma.
\end{proof}

Finally, we are ready to prove that $\mathcal{E}$ exchanges at least $\frac{t^2}{32}$ messages.

\begin{lemma}
\label{lemma:last_lemma}
The message complexity of $\mathcal{A}$ is at least $\frac{t^2}{32}$.
\end{lemma}
\begin{proof}
According to \Cref{lemma:R}, there exists a round $R \in \mathbb{N}$ such that (1) processes from group $A$ decide $1$ in $\mathcal{E}_0^{B(R)}$, and (2) processes from group $A$ decide $0$ in $\mathcal{E}_0^{B(R + 1)}$.
By \Cref{definition:mergeable}, executions $\mathcal{E}_0^{B(R)}$ and $\mathcal{E}_0^{C(R)}$ are mergeable.
As processes from group $A$ decide $1$ in $\mathcal{E}_0^{B(R)}$, \Cref{lemma:either_or} implies that processes from group $A$ decide $1$ in $\mathcal{E}_0^{C(R)}$.
Moreover, executions $\mathcal{E}_0^{B(R + 1)}$ and $\mathcal{E}_0^{C(R)}$ are mergeable according to \Cref{definition:mergeable}.
Thus, by \Cref{lemma:either_or}, all processes from group $A$ decide $1$ in $\mathcal{E}_0^{B(R + 1)}$ (as they do so in $\mathcal{E}_0^{C(R)}$).
This is a contradiction with the fact that processes from group $A$ decide $0$ in $\mathcal{E}_0^{B(R + 1)}$.
Hence, the assumption of $\mathcal{A}$'s message complexity being less than $\frac{t^2}{32}$ must be wrong.
\end{proof}
\begin{figure}[ht]
    \centering
    \begin{minipage}{1\textwidth}
        \centering
        \includegraphics[width=0.9\linewidth]{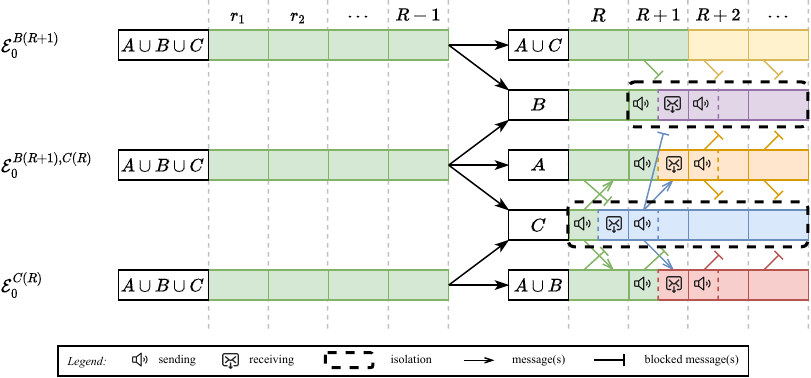}
        \caption{Illustration of \Cref{lemma:either_or} used in the proof of \Cref{lemma:last_lemma}.
        The arrows denoting messages are not exhaustive.
        As in \Cref{figure:helper_isolation_illustration}, the colors represent the local behaviors of processes.
        This picture illustrates why group $A$ is forced to decide the same value in executions $\mathcal{E}_0^{B(R+1)}$ and $\mathcal{E}_0^{C(R)}$.
        Consider the ``merged'' execution $\mathcal{E}_0^{B(R+1), C(R)}$ where $B$ and $C$ are isolated at rounds $R+1$ and $R$, respectively.
        If $A$ decides differently in $\mathcal{E}_0^{B(R+1)}$ (row 1) and $\mathcal{E}_0^{C(R)}$ (row 5), then majorities of $B$ and $C$ decide differently in the $\mathcal{E}_0^{B(R+1), C(R)}$ (rows 2 and 4) due to indistinguishability.
        Group $A$ in $\mathcal{E}_0^{B(R+1), C(R)}$ (row 3) then disagrees with either a majority of $B$ (row 2) or a majority of $C$ (row 4), contradicting \Cref{lemma:lower_bound_helper}.}
        \label{figure:helper_final_lemma_illustration}
    \end{minipage}
\end{figure}

%% file: sections/weak_agreement.tex
\section{Generalization of the Lower Bound} \label{section:weak_consensus}

In this section, we extend the quadratic lower bound proven for weak consensus (see \Cref{section:lower_bound_weak_consensus}) to all non-trivial (without an always-admissible decision) Byzantine agreement problems:

\begin{theorem} \label{theorem:general_lower_bound}
Any algorithm that solves any non-trivial Byzantine agreement problem has $\Omega(t^2)$ message complexity.    
\end{theorem}

To prove the general lower bound (\Cref{theorem:general_lower_bound}), we show that weak consensus is the weakest non-trivial agreement problem.
Namely, we present a zero-message reduction from weak consensus to any (solvable) non-trivial agreement problem.

\subsection{Validity Properties} \label{subsection:validity_definition}

To capture any specific Byzantine agreement problem, we require a generic definition of the validity property.
For that purpose, we reuse the formalism (and nomenclature) of~\cite{validity_podc}.
In brief, a validity property maps the proposals of correct processes into a set of admissible decisions.

Let a \emph{process-proposal} pair be a pair $(p_i, v)$, where $p_i \in \Pi$ is a process and $v \in \mathcal{V}_I$ is a proposal.
Given any process-proposal pair $\mathit{pp} = (p_i, v)$, we denote by $\mathsf{proposal}(\mathit{pp}) = v$ the proposal associated with the pair. 
An \emph{input configuration} is a tuple $\big[ \mathit{pp}_1, \mathit{pp}_2, ..., \mathit{pp}_x \big]$ such that (1) $n - t \leq x \leq n$, and (2) every process-proposal pair is associated with a distinct process.
In a nutshell, an input configuration is an assignment of proposals to (all) correct processes.
For instance, $\big[ (p_1, v_1), (p_4, v_4), (p_5, v_5) \big]$ is an input configuration according to which (1) only processes $p_1$, $p_4$ and $p_5$ are correct, and (2) $p_1$ proposes $v_1$, $p_4$ proposes $v_4$ and $p_5$ proposes $v_5$.

The set of all input configurations is denoted by $\mathcal{I}$.
Moreover, $\mathcal{I}_n \subsetneq \mathcal{I}$ denotes the set of all input configurations with exactly $n$ process-proposals pairs.
Given any input configuration $c \in \mathcal{I}$, $c[i]$ denotes the process-proposal pair associated with the process $p_i$; if such a process-proposal pair does not exist, $c[i] = \bot$.
Moreover, $\pi(c) = \{ p_i \in \Pi \,|\, c[i] \neq \bot \}$ denotes the set of all correct processes according to any input configuration $c \in \mathcal{I}$.

\paragraph{Execution - input configuration correspondence.}
Let $\mathcal{E}$ be any execution of any algorithm $\mathcal{A}$ which exposes the $\mathsf{propose}(\cdot) / \mathsf{decide}(\cdot)$ interface, and let $c \in \mathcal{I}$ be any input configuration.
We say that $\mathcal{E}$ \emph{corresponds to} $c$ (in short, $\mathsf{input\_conf}(\mathcal{E}) = c$) if and only if:
\begin{compactitem}
    \item $\pi(c) = \mathit{Correct}_{\mathcal{A}}(\mathcal{E})$, i.e., the set of processes which are correct in $\mathcal{E}$ is identical to the set of processes which are correct according to $c$; and

    \item for every process $p_i \in \pi(c)$, $p_i$'s proposal in $\mathcal{E}$ is $\mathsf{proposal}(c[i])$.
\end{compactitem}

\paragraph{Satisfying validity.}
A validity property $\mathit{val}$ is a function $\mathit{val}: \mathcal{I} \to 2^{\mathcal{V}_O}$ such that $\mathit{val}(c) \neq \emptyset$, for every input configuration $c \in \mathcal{I}$.
We say that any algorithm $\mathcal{A}$ which exposes the $\mathsf{propose}(\cdot) / \mathsf{decide}(\cdot)$ interface \emph{satisfies} a validity property $\mathit{val}$ if and only if, in any execution $\mathcal{E} \in \mathit{execs}(\mathcal{A})$, no correct process decides any value $v' \notin \mathit{val}\big( \mathsf{input\_conf}(\mathcal{E}) \big)$.
Intuitively, an algorithm satisfies a validity property if correct processes only decide admissible values.

\paragraph{The defining property of Byzantine agreement.}
Observe that an exact definition of validity uniquely defines a \emph{specific} agreement problem.
Indeed, any validity property encodes information about (1) $n$, the total number of processes, (2) $t$, the upper bound on the number of failures, (3) $\mathcal{V}_I$, the set of proposals, and (4) $\mathcal{V}_O$, the set of decisions.
We refer to a specific agreement problem with a validity property $\mathit{val}$ as the ``$\mathit{val}$-agreement'' problem.
Lastly, we recall that the $\mathit{val}$-agreement problem, for some validity property $\mathit{val}$, is \emph{trivial} if and only if there exists an always-admissible value, i.e.,
\[
\exists v' \in \mathcal{V}_O: v' \in \bigcap\limits_{c \in \mathcal{I}} \mathit{val}(c). 
\]

\subsection{Weak Consensus: The Weakest Non-Trivial Byzantine Agreement Problem} \label{subsection:reduction}

In this subsection, we prove that any solution to any non-trivial agreement problem yields, at no additional communication cost, a solution to weak consensus:

\begin{lemma} \label{lemma:reduction_exists}
There exists a zero-message reduction from weak consensus to any solvable non-trivial Byzantine agreement problem.
\end{lemma}

Before presenting the reduction, we introduce the \emph{containment relation}.

\paragraph{Containment relation.}
We define the containment relation (``$\sqsupseteq$'') between input configurations:
\begin{equation*}
    \forall c_1, c_2 \in \mathcal{I}: c_1 \sqsupseteq c_2 \iff (\pi(c_1) \supseteq \pi(c_2)) \land (\forall p_i \in \pi(c_2): c_1[i] = c_2[i]). 
\end{equation*}
Intuitively, $c_1$ contains $c_2$ if and only if (1) each process in $c_2$ belongs to $c_1$, and (2) for each process in $c_2$, its proposals in $c_1$ and $c_2$ are identical.
For example, when $n = 3$ and $t = 1$, $\big[ (p_1, v_1), (p_2, v_2), (p_3, v_3) \big]$ contains $\big[ (p_1, v_1), (p_3, v_3) \big]$, but it does not contain $\big[ (p_1, v_1), (p_3, v_3' \neq v_3) \big]$.
Note that the containment relation is reflexive (for every $c \in \mathcal{I}$, $c \sqsupseteq c$).
For any input configuration $c \in \mathcal{I}$, we define its \emph{containment set} $\mathit{Cnt}(c)$ as the set of all input configurations which $c$ contains:
\begin{equation*}
    \mathit{Cnt}(c) = \{ c' \in \mathcal{I} \,|\, c \sqsupseteq c' \}.
\end{equation*}

The following lemma proves that, in any execution that corresponds to some input configuration $c$, if any agreement algorithm decides some value $v'$, then $v'$ must be admissible according to all input configurations $c$ contains.
Otherwise, the same scenario can correspond to \emph{another} input configuration for which $v'$ is not admissible, thus violating the considered validity property.
A formal proof of the following lemma is relegated to \Cref{section:containment_proof}.

\begin{restatable}{lemma}{containment}
\label{lemma:canonical_containment}
Let $\mathcal{A}$ be any algorithm that solves the $\mathit{val}$-agreement problem, for any validity property $\mathit{val}$.
Let $\mathcal{E}$ be any (potentially infinite) execution of $\mathcal{A}$, and let $c = \mathsf{input\_conf}(\mathcal{E})$, for some input configuration $c$.
If a correct process decides a value $v' \in \mathcal{V}_O$ in $\mathcal{E}$, then $v' \in \bigcap\limits_{c' \in \mathit{Cnt}(c)} \mathit{val}(c')$.
\end{restatable}

\paragraph{Reduction.}
We fix any solvable non-trivial agreement problem $\mathcal{P}$, and any algorithm $\mathcal{A}$ which solves $\mathcal{P}$.
Let $\mathit{val}$ denote the specific validity property of $\mathcal{P}$.
Moreover, we fix the following notation:

\begin{table} [h]
    \begin{tabular}{|l|m{8.5cm}|}
    \multicolumn{1}{l}{\emph{Notation}} & \multicolumn{1}{l}{\emph{Definition \& commentary}} \\ 
    \hline
    \hline
        $c_0 \in \mathcal{I}_n$ & Any input configuration (of $\mathcal{P}$) according to which all processes are correct ($\pi(c_0) = \Pi$).
         \\ \hline
        $\mathcal{E}_0 \in \mathit{execs}(\mathcal{A})$ & The infinite execution of $\mathcal{A}$ such that $\mathsf{input\_conf}(\mathcal{E}_0) = c_0$. \\ \hline
        $v_0' \in \mathcal{V}_O$ & The value decided in $\mathcal{E}_0$.
        Note that such a value exists as $\mathcal{A}$ satisfies \emph{Termination} and \emph{Agreement}. \\ \hline
        $c_1^* \in \mathcal{I}$ & Any input configuration (of $\mathcal{P}$) such that $v_0' \notin \mathit{val}(c_1^*)$.
        Note that such an input configuration exists as $\mathcal{P}$ is non-trivial. \\ \hline
        $c_1 \in \mathcal{I}_n$ & Any input configuration (of $\mathcal{P}$) such that (1) $c_1 \sqsupseteq c_1^*$, and (2) all processes are correct according to $c_1$ ($\pi(c_1) = \Pi$).
        Note that such an input configuration exists as the containment condition is reflexive. \\ \hline
        $\mathcal{E}_1 \in \mathit{execs}(\mathcal{A})$ & The infinite execution of $\mathcal{A}$ such that $\mathsf{input\_conf}(\mathcal{E}_1) = c_1$. \\ \hline
        $v_1' \in \mathcal{V}_O$ & The value decided in $\mathcal{E}_1$.
        Note that such a value exists as $\mathcal{A}$ satisfies \emph{Termination} and \emph{Agreement}.
        Crucially, as $c_1 \sqsupseteq c_1^*$ and $v_0' \notin \mathit{val}(c_1^*)$, \Cref{lemma:canonical_containment} proves that $v_1' \neq v_0'$. \\\hline 
    \end{tabular}
    \caption{Notation table for the reduction.}
    \label{table:notation_reduction}
\end{table}

The reduction from weak consensus to $\mathcal{P}$ is presented in \Cref{algorithm:reduction}.
Our crucial observation is that $\mathcal{A}$, the fixed algorithm solving $\mathcal{P}$, decides \emph{different} values in $\mathcal{E}_0$ and $\mathcal{E}_1$: by \Cref{lemma:canonical_containment}, the value $v_1'$ decided in $\mathcal{E}_1$ is admissible according to $c_1^*$ (as $c_1 \sqsupseteq c_1^*$), which implies that $v_1' \neq v_0'$.
We utilize the aforementioned fact to distinguish (1) the fully correct execution $\mathcal{E}_0^w$ of weak consensus where all processes propose $0$, and (2) the fully correct execution $\mathcal{E}_1^w$ of weak consensus where all processes propose $1$.
Namely, our reduction works as follows:
If a correct process $p_i$ proposes $0$ (resp., $1$) to weak consensus, $p_i$ proposes its proposal from the input configuration $c_0$ (resp., $c_1$) to the underlying algorithm $\mathcal{A}$.
Moreover, if $p_i$ decides $v_0'$ from $\mathcal{A}$, $p_i$ decides $0$ from weak consensus; otherwise, $p_i$ decides $1$ from weak consensus.
Thus, if all processes are correct and propose $0$ (resp., $1$) to weak consensus, $\mathcal{A}$ \emph{necessarily} decides $v_0'$ (resp., $v_1' \neq v_0'$), which then implies that all correct processes decide $0$ (resp., $1$) from weak consensus, thus satisfying \emph{Weak Validity}.
The correctness of our reduction is proven in \Cref{section:reduction_proof}.

\begin{algorithm}
\caption{Reduction from weak consensus to $\mathcal{P}$: Pseudocode for process $p_i$}
\footnotesize
\label{algorithm:reduction}
\begin{algorithmic} [1]
\State \textbf{Uses:}
\State \hskip2em $\mathcal{A}$, an algorithm solving the non-trivial agreement problem $\mathcal{P}$

\medskip
\State \textbf{upon} $\mathsf{propose}(b \in \{0, 1\})$: \label{line:wbc_propose}
\State \hskip2em \textbf{if} $b = 0$:
\State \hskip4em \textbf{invoke} $\mathcal{A}.\mathsf{propose}\big( \mathsf{proposal}(c_0[i]) \big)$ \label{line:wbc_propose_c1}
\State \hskip2em \textbf{else:}
\State \hskip4em \textbf{invoke} $\mathcal{A}.\mathsf{propose}\big( \mathsf{proposal}(c_1[i]) \big)$ \label{line:wbc_propose_c2}

\medskip
\State \textbf{upon} $\mathcal{A}.\mathsf{decide}(\mathit{decision} \in \mathcal{V}_O)$:
\State \hskip2em \textbf{if} $\mathit{decision} = v_0'$: \label{line:wbc_decide_v0}
\State \hskip4em \textbf{trigger} $\mathsf{decide}(0)$ \label{line:wbc_decide_0}
\State \hskip2em \textbf{else} 
\State \hskip4em \textbf{trigger} $\mathsf{decide}(1)$ \label{line:wbc_decide_1}
\end{algorithmic}
\end{algorithm}

Importantly, our reduction proves the general quadratic lower bound (\Cref{theorem:general_lower_bound}).
Indeed, if there was a sub-quadratic algorithm $\mathcal{A}$ which solves any non-trivial Byzantine agreement problem, the introduced reduction would yield a sub-quadratic weak consensus algorithm, thus contradicting the quadratic lower bound for weak consensus (proven in \Cref{section:lower_bound_weak_consensus}).

\subsection{On the Lower Bound for the Blockchain-Specific Agreement Problem}

At the heart of today's blockchain systems lies an agreement problem that requires the decided value to satisfy a globally verifiable condition.
Concretely, modern blockchain systems satisfy the following validity property:
\begin{compactitem}
    \item \emph{External Validity}~\cite{Cachin2001}: If a correct process decides a value $v'$, then $\mathsf{valid}(v') = \mathit{true}$, where $\mathsf{valid}(\cdot)$ is a globally verifiable predicate. 
\end{compactitem}
This subsection underlines that the general quadratic lower bound (\Cref{theorem:general_lower_bound}) extends to all ``reasonable'' agreement problems with \emph{External Validity}.

\emph{External Validity} emerged as the validity property of blockchain systems because stronger notions of validity have limited applicability in this setting.
For example, consider \emph{Strong Validity} which guarantees only that, if all correct processes propose the same value, that value must be decided.
Whenever correct processes do not propose the same value, \emph{Strong Validity} provides no guarantees, e.g., a value proposed by a faulty process can be decided.
In a blockchain setting, it will rarely be the case that all correct validators (i.e., processes that operate the blockchain) construct and propose an identical block with the clients' pending transactions.
Hence, the chain could be comprised of ``faulty blocks'', thus allowing faulty validators to commit invalid (e.g., incorrectly signed) transactions.
\emph{External Validity} eliminates this problem by allowing only valid blocks to be committed.

As mentioned in~\cite{validity_podc}, the formalism we use for defining validity properties (see \Cref{subsection:validity_definition}) is not suitable for expressing \emph{External Validity}.
Namely, the formalism would technically classify \emph{External Validity} as a trivial validity property since any fixed valid value is admissible according to \emph{every} input configuration.
However, in practice, the agreement problem with \emph{External Validity} does not allow for a trivial solution in the blockchain setting.
For example, the fact that some transaction $\mathit{tx}$, which is \emph{correctly signed} by some client, is valid does not mean that validators can always decide $\mathit{tx}$.
Indeed, for a validator to decide $\mathit{tx}$, it needs to first \emph{learn} about $\mathit{tx}$ (otherwise, cryptographic hardness assumptions on signatures would break).
Therefore, validators cannot decide $\mathit{tx}$ ``on their own'', which precludes a trivial solution to agreement problems with \emph{External Validity}.

Nonetheless, our quadratic lower bound applies to \emph{any} algorithm $\mathcal{A}$ which solves Byzantine agreement with \emph{External Validity} as long as the algorithm has two fully correct executions with different decisions.
Indeed, if $\mathcal{A}$ has two fully correct infinite executions $\mathcal{E}_0$ and $\mathcal{E}_1$ that decide different values, \Cref{algorithm:reduction} (see \Cref{subsection:reduction}) solves weak consensus using $\mathcal{A}$ by employing $c_0 = \mathsf{input\_conf}(\mathcal{E}_0)$ (line~\ref{line:wbc_propose_c1} of \Cref{algorithm:reduction}) and $c_1 = \mathsf{input\_conf}(\mathcal{E}_1)$ (line~\ref{line:wbc_propose_c2} of \Cref{algorithm:reduction}).
To the best of our knowledge, every known agreement algorithm with \emph{External Validity} (e.g.,~\cite{yin2019hotstuff,BKM19,CGL18,lewis2022quadratic}) has different fully correct executions in which different values are decided.
Concretely, it is ensured that, if all processes are correct and they all propose the same value, that value will be decided.\footnote{In other words, all these agreement algorithms satisfy \emph{both} \emph{External Validity} and \emph{Weak Validity}.}

\begin{corollary}
Let $\mathcal{A}$ be any algorithm that solves Byzantine agreement with \emph{External Validity}.
Moreover, let there exist two executions $\mathcal{E}_0$ and $\mathcal{E}_1$ of $\mathcal{A}$ such that (1) all processes are correct in both $\mathcal{E}_0$ and $\mathcal{E}_1$, (2) some value $v_0'$ is decided in $\mathcal{E}_0$, and (3) some value $v_1' \neq v_0'$ is decided in $\mathcal{E}_1$.
Then, $\mathcal{A}$ has at least $\frac{t^2}{32}$ message complexity.
\end{corollary}

%% file: sections/solvability.tex
\section{Solvability of Byzantine Agreement Problems} \label{section:solvability}

In this section, we observe that a deeper study of the containment relation (introduced in \Cref{subsection:reduction}) enables us to deduce which Byzantine agreement problems are solvable in synchrony.
Concretely, we introduce the general solvability theorem, which unifies all previous results on the synchronous solvability of Byzantine agreement problems (e.g.,~\cite{LSP82,FLM85,lynch1996distributed,dolev1983authenticated,abraham2022authenticated,fitzi2003efficient}).

\subsection{Authenticated \& Unauthenticated Algorithms}

When it comes to the solvability of Byzantine agreement problems in synchrony, authentication makes a significant difference.
For instance,~\cite{dolev1983authenticated} proved that authenticated Byzantine broadcast can tolerate any number $t < n$ of corrupted processes, whereas~\cite{LSP82} showed that unauthenticated Byzantine broadcast cannot be solved unless $n > 3t$.
We thus distinguish two types of algorithms:
\begin{compactitem}
    \item \emph{Authenticated algorithms}, which allow processes to sign their messages in a way that prevents their signature from being forged by any other process \cite{Canetti04}.

    \item \emph{Unauthenticated algorithms}, which do not provide any mechanism for signatures.
    (Note that the receiver of a message knows the identity of its sender.)
\end{compactitem}
A Byzantine agreement problem $\mathcal{P}$ is \emph{authenticated-solvable} (resp., \emph{unauthenticated-solvable}) if and only if there exists an authenticated (resp., unauthenticated) algorithm which solves $\mathcal{P}$.\footnote{Recall that the exact specification of $\mathcal{P}$ (concretely, $\mathcal{P}$'s validity property) encodes the resilience of $\mathcal{P}$.}

\paragraph{Remark about unauthenticated algorithms.}
This section assumes that unauthenticated algorithms confront the adversary that is able to simulate other processes.
In other words, we do not assume the resource-restricted paradigm \cite{Garay2020RRC}, where the adversary's capability to simulate other processes can be restricted assuming a per-process bounded rate of cryptographic puzzle-solving capability with no bound on the number of corruptions and without any setup (i.e., without any authentication mechanism)~\cite{Andrychowicz2015,Katz2014}.

\subsection{General Solvability Theorem}

Before presenting our solvability theorem, we define its key component -- the \emph{containment condition}.

\begin{definition} [Containment condition] \label{definition:reduction}
A non-trivial agreement problem $\mathcal{P}$ with some validity property $\mathit{val}$ satisfies the \emph{containment condition} ($\mathcal{CC}$, in short) if and only if there exists a Turing-computable function $\Gamma: \mathcal{I} \to \mathcal{V}_O$ such that:
\begin{equation*}
    \forall c \in \mathcal{I}: \Gamma(c) \in \bigcap\limits_{c' \in \mathit{Cnt}(c)} \mathit{val}(c').
\end{equation*}
\end{definition}
Intuitively, a non-trivial agreement problem satisfies $\mathcal{CC}$ if and only if there exists a finite procedure which, for every input configuration $c \in \mathcal{I}$, returns a value that is admissible according to \emph{all} input configurations to which $c$ reduces.

We are now ready to introduce the general solvability theorem:

\begin{theorem} [General solvability theorem] \label{theorem:solvability}
A non-trivial Byzantine agreement problem $\mathcal{P}$ is:
\begin{compactitem}
    \item authenticated-solvable if and only if $\mathcal{P}$ satisfies $\mathcal{CC}$; and 

    \item unauthenticated-solvable if and only if (1) $\mathcal{P}$ satisfies $\mathcal{CC}$, and (2) $n > 3t$.
\end{compactitem}
\end{theorem}


To prove the general solvability theorem (\Cref{theorem:solvability}), we show the following three results:
\begin{compactitem}
    \item \emph{Necessity of $\mathcal{CC}$:} If a non-trivial Byzantine agreement problem $\mathcal{P}$ is authenticated- or unauthenticated-solvable, then $\mathcal{P}$ satisfies $\mathcal{CC}$.

    \item \emph{Sufficiency of $\mathcal{CC}$:} If a non-trivial Byzantine agreement problem $\mathcal{P}$ satisfies $\mathcal{CC}$ (resp., satisfies $\mathcal{CC}$ and $n > 3t$), then $\mathcal{P}$ is authenticated-solvable (resp., unauthenticated-solvable).

    \item \emph{Unauthenticated triviality when $n \leq 3t$:} If a Byzantine agreement problem $\mathcal{P}$ is unauthenticated-solvable with $n \leq 3t$, then $\mathcal{P}$ is trivial.
\end{compactitem}


\subsubsection{Necessity of $\mathcal{CC}$}

The necessity of $\mathcal{CC}$ for solvable non-trivial agreement problems follows directly from \Cref{lemma:canonical_containment}:

\begin{lemma} \label{lemma:necessity_cc}
If a non-trivial Byzantine agreement problem $\mathcal{P}$ is authenticated- or unauthenticated-solvable, then $\mathcal{P}$ satisfies $\mathcal{CC}$.
\end{lemma}
\begin{proof}
Let $\mathcal{P}$ be any authenticated- or unauthenticated-solvable non-trivial Byzantine agreement problem.
Let $\mathit{val}$ denote the validity property of $\mathcal{P}$.
As $\mathcal{P}$ is solvable, there exists an (authenticated or unauthenticated) algorithm $\mathcal{A}$ which solves it.

Let us fix any input configuration $c \in \mathcal{I}$. 
Consider any infinite execution $\mathcal{E} \in \mathit{execs}(\mathcal{A})$ such that $c = \mathsf{input\_conf}(\mathcal{E})$.
As $\mathcal{E}$ is infinite, some correct process decides (in finitely many rounds) some value $v' \in \mathcal{V}_O$ (to satisfy \emph{Termination}).
Due to \Cref{lemma:canonical_containment}, $v' \in \bigcap\limits_{c' \in \mathit{Cnt}(c)} \mathit{val}(c')$.
Thus, $\Gamma(c)$ is defined (as $\Gamma(c) = v'$) and is Turing-computable ($\mathcal{A}$ computes it in $\mathcal{E}$).
Hence, $\mathcal{P}$ satisfies $\mathcal{CC}$.
\end{proof}

\subsubsection{Sufficiency of $\mathcal{CC}$}

Let us start by recalling interactive consistency, a specific Byzantine agreement problem.
In interactive consistency, each process proposes its value, and processes decide vectors of $n$ elements, one for each process (i.e., $\mathcal{V}_O = \mathcal{I}_n$).
Besides \emph{Termination} and \emph{Agreement}, interactive consistency requires the following validity property to hold:
\begin{compactitem}
\item \emph{IC-Validity}: Let $V$ denote the vector decided by a correct process. If a correct process $p_i$ proposed a value $v$, then $V[i] = v$. 
\end{compactitem}
Using our formalism, \emph{IC-Validity} can be expressed as $\text{\emph{IC-Validity}}(c) = \{c' \in \mathcal{I}_n \,|\, c' \sqsupseteq c\}$.
Importantly, interactive consistency is authenticated-solvable for any $n$ and any $t \in [1, n - 1]$~\cite{dolev1983authenticated}.
On the other hand, interactive consistency is unauthenticated-solvable if $n > 3t$~\cite{LSP82,FLM85}.

To prove the sufficiency of $\mathcal{CC}$, we prove that any non-trivial Byzantine agreement problem that satisfies $\mathcal{CC}$ can be reduced to interactive consistency at no resilience penalty.

\begin{lemma}
If a non-trivial Byzantine agreement problem $\mathcal{P}$ satisfies $\mathcal{CC}$ (resp., satisfies $\mathcal{CC}$ and $n > 3t$), then $\mathcal{P}$ is authenticated-solvable (resp., unauthenticated-solvable).
\end{lemma}
\begin{proof}
To prove the lemma, we design a reduction from $\mathcal{P}$ to interactive consistency (\Cref{algorithm:generic_algorithm}).
Our reduction is comprised of two steps:
(1) When a correct process proposes to $\mathcal{P}$ (line~\ref{line:propose}), the process forwards its proposal to the underlying interactive consistency algorithm (line~\ref{line:propose_ic}).
(2) Once a correct process decides a vector $\mathit{vec}$ of $n$ proposals from interactive consistency (line~\ref{line:decide_ic}), the process decides $\Gamma(\mathit{vec})$ from $\mathcal{P}$ (line~\ref{line:decide}).
\emph{Termination} and \emph{Agreement} of the reduction algorithm follow directly from \emph{Termination} and \emph{Agreement} of interactive consistency, respectively.
Finally, let us prove that the reduction algorithm satisfies the specific validity property $\mathit{val}$ of $\mathcal{P}$.
Consider any specific execution $\mathcal{E}$ of the reduction algorithm such that $\mathsf{input\_conf}(\mathcal{E}) = c$, for some input configuration $c \in \mathcal{I}$.
Let $\mathit{vec} \in \mathcal{I}_n$ denote the vector which a correct process decides from the underlying interactive consistency algorithm (line~\ref{line:decide_ic}).
\emph{IC-Validity} ensures that $\mathit{vec} \sqsupseteq c$ as, for every correct process $p_i$, $\mathit{vec}[i] = \mathsf{proposal}(c[i])$.
As $\mathcal{P}$ satisfies $\mathcal{CC}$, $\Gamma(\mathit{vec}) \in \mathit{val}(c)$, which proves that the reduction algorithm satisfies $\mathit{val}$.

As interactive consistency is authenticated-solvable for any $n$ and any $t \in [1, n - 1]$~\cite{dolev1983authenticated}, a non-trivial Byzantine agreement problem $\mathcal{P}$ which satisfies $\mathcal{CC}$ is authenticated-solvable.
Similarly, as interactive consistency is unauthenticated-solvable if $n > 3t$~\cite{LSP82,FLM85}, a non-trivial Byzantine agreement problem $\mathcal{P}$ which satisfies $\mathcal{CC}$ with $n > 3t$ is unauthenticated-solvable.
\end{proof}

\begin{algorithm}
\caption{Reduction from $\mathcal{P}$ to interactive consistency: Pseudocode for process $p_i$} \label{algorithm:generic_algorithm}
\footnotesize
\begin{algorithmic} [1]
\State \textbf{Uses:}
\State \hskip2em Interactive consistency, \textbf{instance} $\mathcal{IC}$

\medskip
\State \textbf{upon} $\mathsf{propose}(v \in \mathcal{V}_I)$: \label{line:propose}
\State \hskip2em \textbf{invoke} $\mathcal{IC}.\mathsf{propose}(v)$ \label{line:propose_ic}

\medskip
\State \textbf{upon} $\mathcal{IC}.\mathsf{decide}(\mathit{vec} \in \mathcal{I}_n)$: \BlueComment{processes decide input configurations with $n$ process-proposal pairs} \label{line:decide_ic}
\State \hskip2em \textbf{trigger} $\mathsf{decide}\big( \Gamma(\mathit{vec}) \big)$ \label{line:decide}
\end{algorithmic}
\end{algorithm}

\subsubsection{Unauthenticated triviality when $n \leq 3t$}

We prove that any agreement problem that is unauthenticated-solvable with $n \leq 3t$ is trivial by contradiction.
Namely, if there existed a non-trivial agreement problem $\mathcal{P}$ that is unauthenticated-solvable with $n \leq 3t$, the reduction from weak consensus to $\mathcal{P}$ presented in \Cref{algorithm:reduction} would yield an unauthenticated weak consensus algorithm with $n \leq 3t$, which is known to be impossible~\cite{FLM85}.

\begin{lemma} \label{lemma:triviality}
If a Byzantine agreement problem $\mathcal{P}$ is unauthenticated-solvable with $n \leq 3t$, then $\mathcal{P}$ is trivial.
\end{lemma}
\begin{proof}
By contradiction, let $\mathcal{P}$ be non-trivial.
As weak consensus can be reduced to any solvable non-trivial agreement problem at no resilience penalty (see \Cref{algorithm:reduction}), weak consensus is unauthenticated-solvable with $n \leq 3t$.
This is a contradiction with the fact that weak consensus is unauthenticated-solvable only if $n > 3t$~\cite{FLM85}.
\end{proof}


\subsection{General Solvability Theorem: Application to Strong Consensus}

Here, we show how the general solvability theorem (\Cref{theorem:solvability}) can be applied with the example of strong consensus.
(Recall that strong consensus satisfies \emph{Strong Validity}: if all correct processes propose the same value, that value must be decided.)
Namely, it is known that strong consensus is authenticated-solvable only if $n > 2t$~\cite{abraham2022authenticated}.
The general solvability theorem enables us to obtain another proof of this claim.

\begin{theorem} [Proven in~\cite{abraham2022authenticated}]
Strong consensus is authenticated-solvable only if $n > 2t$.
\end{theorem}
\begin{proof}
To prove the theorem, we show that strong consensus does not satisfy $\mathcal{CC}$ with $n \leq 2t$.
Without loss of generality, let $n = 2t$ and let $\mathcal{V}_I = \mathcal{V}_O = \{0, 1\}$.
Consider the input configuration $c \in \mathcal{I}_n$ such that (1) for every $i \in [1, t]$, $\mathsf{proposal}(c[i]) = 0$, and (2) for every $i \in [t + 1, n]$, $\mathsf{proposal}(c[i]) = 1$.
That is, the proposal of the first $t$ processes is $0$, whereas the proposal of the other processes is $1$.
Note that both $0$ and $1$ are admissible according to $c$.
Importantly, $c$ contains $c_0 \in \mathcal{I}_t$, where $\pi(c_0) = \{p_1, p_2, ..., p_t\}$ and $\mathsf{proposal}(c_0[i]) = 0$, for every $i \in [1, t]$.
Similarly, $c$ contains $c_1 \in \mathcal{I}_t$, where $\pi(c_0) = \{p_{t+1}, p_{t + 2}, ..., p_n\}$ and $\mathsf{proposal}(c_0[i]) = 1$, for every $i \in [t + 1, n]$.
According to $c_0$ (resp., $c_1$), only $0$ (resp., $1$) is admissible.
Hence, strong consensus with $n \leq 2t$ does not satisfy $\mathcal{CC}$ as $c$ contains two input configurations (namely, $c_0$ and $c_1$) which do not have a common admissible value.
\end{proof}

%% file: sections/related.tex
\section{Related Work}
\label{section:related_work}

\paragraph{Reductions and equivalences between Byzantine agreement problems.}
Interactive consistency can be reduced to $n$ (parallel) instances of Byzantine broadcast~\cite{Nayak2020a}.
In the honest-majority setting ($t < \frac{n}{2}$), Byzantine broadcast and strong consensus are computationally equivalent~\cite{lynch1996distributed,AW04}.
Moreover, Byzantine broadcast can be reduced to strong consensus with only $O(n)$ additional exchanged messages~\cite{lynch1996distributed,AW04}.
Furthermore, it is known that weak consensus is reducible to strong consensus (and, thus, to Byzantine broadcast)~\cite{lynch1996distributed,AW04}.

\paragraph{Deterministic Byzantine agreement in synchrony.}
In their seminal paper,~\cite{dolev1985bounds} established a quadratic lower bound on message complexity of deterministic Byzantine broadcast (and, consequently, strong consensus). It is shown that, in the authenticated setting (with idealized digital signatures \cite{Canetti04}), deterministic Byzantine broadcast algorithms must exchange $\Omega(nt)$ signatures and $\Omega(n + t^2)$ messages.
Similarly, their proof shows that, in the unauthenticated setting, there exists an execution with $\Omega(nt)$ exchanged messages. 
The $\Omega(nt)$ bound on exchanged signatures is proven to be tight when $t < \frac{n}{2} - O(1)$ and $t \in \Theta(n)$~\cite{Momose2021}.
Additionally,~\cite{berman1992bit} proved that the $\Omega(nt)$ bound on message complexity in the unauthenticated setting is tight when $t \in \Theta(n)$.
The $\Omega(n + t^2)$ bound on message complexity in the authenticated setting has recently been proven to be tight~\cite{Chlebus23}.
A quadratic lower bound on the message complexity of binary crusader broadcast, a problem in which disagreements are sometimes allowed, has also been shown in~\cite{AbrahamStern22}.
Lower bounds on other relevant metrics, such as resilience, network connectivity, or latency, have also been established~\cite{FLM85,dolev1983authenticated,dolev2013early}.

By employing threshold signatures~\cite{Shoup00}, which extend beyond the idealized authenticated model, the word complexity of $O(n(f+1))$, where $f \leq t < \frac{n}{2}$ represents the actual number of failures and a word contains a constant number of values and signatures, can be achieved for Byzantine agreement with \emph{External Validity}~\cite{spiegelman2020search} and Byzantine broadcast~\cite{cohen2023make,strong} by utilizing the algorithm of~\cite{Momose2021}. 
Additionally, an amortized cost of $O(n)$ is attainable in multi-shot Byzantine broadcast~\cite{wan2023amortized}. 
Amortization is similarly possible with long inputs~\cite{Chen2021,Nayak2020a}.
In the dishonest-majority setting (with $t \geq \frac{n}{2}$), the most efficient broadcast constructions are based on the deterministic broadcast protocol of~\cite{dolev1985bounds} with a cubic message complexity.  

\paragraph{Randomized Byzantine agreement in synchrony.}
 Even with randomization, no Byzantine broadcast algorithm can achieve sub-quadratic expected message complexity against a strongly rushing adaptive adversary equipped with after-the-fact message removal capabilities~\cite{Abraham2019c}. 
However, designing randomized synchronous Byzantine agreement algorithms with sub-quadratic expected message complexity is possible against a weaker adversary. 
In certain models, such as those with a static adversary~\cite{Boyle2021,King2011a} or with (only) private channels~\cite{King2011}, algorithms with a sub-linear number of messages (or bits) sent per correct process can be designed~\cite{Gelles,Gelles23,King2011a,King2009,King2011,Boyle2021,Boyle2018b}.

When the adversary is adaptive (without after-the-fact message removal capabilities) and computationally bounded, there exist Byzantine agreement algorithms~\cite{Chen2019,Abraham2019c,RambaudBootstrapping} which achieve both sub-quadratic (but unbalanced) communication and constant latency in expectation by relying on a verifiable random function (VRF)~\cite{DBLP:conf/focs/MicaliRV99}.
It has been shown that, in the idealized authenticated setting~\cite{Canetti04} (which is strictly weaker than bare or bulletin-board PKI~\cite{Canetti00,Boyle2021,RambaudBootstrapping}), in the presence of a rushing adaptive adversary, no randomized protocol can achieve a sub-quadratic expected communication complexity in the synchronous multi-cast model, where a sent message is necessarily sent to all the processes~\cite{RambaudBootstrapping}.

A VRF setup and a sub-quadratic binary strong consensus algorithm were shown to yield a $O(n \ell + n \mathit{poly}(\kappa))$ bit complexity, where $\ell$ is the proposal size and $\kappa$ is the security parameter, for solving strong consensus with long inputs~\cite{Bhangale2022}.
State-of-the-art algorithms for interactive consistency with long inputs (of size $\ell$) yield the bit complexity of $O(n^2\ell + n^2 \kappa^3)$~\cite{Bhangale2022} or $O(n^2 \ell + n^3\kappa)$~\cite{Abraham2023a}. 

In the dishonest-majority setting,~\cite{Blum2023} establishes new lower bounds on the expected message complexity for Byzantine broadcast: no (randomized) algorithm can achieve sub-quadratic message complexity with only $O(1)$ correct processes.
The algorithm of~\cite{Chan2020} achieves the bit complexity of $O(n^2\kappa^2)$ for binary Byzantine broadcast. 
~\cite{Tsimos2022} proves that an $\tilde{O}(n^2\mathit{poly}(\kappa))$ bit complexity can be achieved for binary interactive consistency.
Randomization additionally helps in circumventing the Dolev-Strong lower bound~\cite{dolev1983authenticated} which states that $t+1$ rounds are necessary in the worst case to deterministically solve Byzantine broadcast~\cite{dolev1983authenticated}.
While using randomization for circumventing the Dolev-Strong lower bound is well-established for the honest-majority setting~\cite{KatzKoo2009,abraham2019synchronous,Abraham2019c,RambaudBootstrapping}, recent findings have proven that the same approach can be utilized even in the presence of a dishonest majority~\cite{Wan2020,Wan2020a,Chan2020}.

\paragraph{Byzantine agreement in partial synchrony and asynchrony.}
The worst-case complexity of all Byzantine agreement problems in partial synchrony was studied in~\cite{validity_podc} where it was proven that any specific Byzantine agreement problem requires $\Theta(n^2)$ exchanged messages (after the network has stabilized) in the worst case.
Prior to~\cite{validity_podc}, it was shown that there exist deterministic algorithms, building on top of threshold signatures and HotStuff~\cite{YMR19}, which achieve $O(n^2)$ word complexity for strong consensus~\cite{lewis2022quadratic,civit2022byzantine}.
Recently,~\cite{everyBitCounts} proved that vector consensus, a Byzantine agreement problem in which processes agree on the proposals of $n - t$ processes, can be solved with $O(n^{2.5})$ or $O(n^2)$ words (when employing STARK proofs~\cite{Ben-Sasson_stark}).
In the randomized paradigm, there exist VRF-based sub-quadratic Byzantine agreement protocols~\cite{Chen2019,Abraham2019c,RambaudBootstrapping,Sheng22}.
Moreover, it is possible to achieve $O(n\ell + n \mathit{poly}(\kappa))$ bit complexity for strong consensus with long inputs of size $\ell$~\cite{Bhangale2022}.
Furthermore, reaching the communication complexity of $O(n\ell + n^2\kappa)$ for validated asynchronous Byzantine agreement was proven to be possible: \cite{LL0W20} and \cite{Nayak2020a} achieve the aforementioned bound by extending the VABA protocol of~\cite{abraham2019asymptotically}.
With some additional assumptions (e.g., private setup or delayed adversary), it is possible to design a sub-quadratic asynchronous Byzantine agreement algorithm~\cite{Blum2020,CKS20}.
A generic transformation proposed in~\cite{Bhangale2022} produces, on top of any asynchronous sub-quadratic Byzantine agreement algorithm, an asynchronous solution with $O(n\ell + n\mathit{poly}(\kappa))$ bit complexity.

%% file: sections/conclusion.tex
\section{Concluding Remarks} \label{section:conclusion}

We study in this paper the necessary worst-case communication cost for all Byzantine agreement problems. We show that any (deterministic) solution to any solvable non-trivial Byzantine agreement problem exchanges $\Omega(t^2)$ messages in the worst-case.
We prove the general lower bound in two steps: 
(1) we show that weak consensus~\cite{yin2019hotstuff,lewis2022quadratic,civit2022byzantine,BKM19} requires $\Omega(t^2)$ exchanged messages even in synchrony;
(2) we design a reduction from weak consensus to any solvable non-trivial Byzantine agreement problem, thus generalizing the $\Omega(t^2)$ bound.
Interestingly, our reduction allows us to determine a general result about the synchronous solvability of Byzantine agreement, thus demarcating the entire landscape of solvable (and unsolvable) variants of the problem.

We plan on extending our results to the randomized setting.
Concretely, 
the goal is to study the cost of solving randomized Byzantine agreement against an adaptive adversary, with after-the-fact message removal capabilities~\cite{Abraham2019c} or the ability to access the internal states of all processes~\cite{DBLP:conf/soda/HuangPZ23}.
It would also be interesting to extend our work to problems which do not require agreement (e.g., approximate~\cite{AbrahamAD04,MendesH13,GhineaLW22,ghinea2023multidimensional} or $k$-set~\cite{BouzidIR16,Delporte-Gallet20,Delporte-Gallet22,lynch1996distributed} agreement).
Finally, improving the known upper bounds on the cost of solving Byzantine agreement problems constitutes another important research direction.

%% file: sections/formal_proof_lower_bound.tex
\section{Proofs of Lemmas~\ref{lemma:lower_bound_helper} and~\ref{lemma:either_or}} \label{section:lower_bound_formal}

In this section, we formally prove \cref{lemma:lower_bound_helper,lemma:either_or}.
Recall that, in \Cref{section:lower_bound_weak_consensus}, we fix a weak consensus algorithm $\mathcal{A}$ which tolerates up to $t < n$ omission failures.
First, we introduce the computational model (\Cref{subsection:computational_model}).
Then, we show preliminary lemmas (\Cref{subsection:prelimariny_lemmas_appendix}) required for proving \cref{lemma:lower_bound_helper,lemma:either_or} (\Cref{subsection:proofs_appendix}).

\subsection{Computational Model} \label{subsection:computational_model}

Without loss of generality, we assume that each process sends at most one message to any specific process in a single round.
That is, $\mathcal{A}$ instructs no process $p_i \in \Pi$ to send two (or more) messages to any specific process $p_j \in \Pi$ in a single round.
Moreover, we assume that no process sends messages to itself.

\subsubsection{Messages.}
Let $\mathcal{M}$ denote the set of messages.
Each message $m \in \mathcal{M}$ encodes the following:
\begin{compactitem}
    \item the sender of $m$ (denoted by $m.\mathsf{sender} \in \Pi$); and

    \item the receiver of $m$ (denoted by $m.\mathsf{receiver} \in \Pi$); and

    \item the round (denoted by $m.\mathsf{round} \in \mathbb{N}$).
\end{compactitem}
Due to our assumption that only one message is sent in a round from any specific process to any other specific process, no message $m$ is sent more than once in any execution of $\mathcal{A}$.

\subsubsection{States.}
Let $\mathcal{S}$ denote the set of states.
Each state $s \in \mathcal{S}$ encodes the following:
\begin{compactitem}
    \item the process associated with $s$ (denoted by $s.\mathsf{process} \in \Pi$); and

    \item the round associated with $s$ (denoted by $s.\mathsf{round} \in \mathbb{N}$); and

    \item the proposal-bit associated with $s$ (denoted by $s.\mathsf{proposal} \in \{0, 1\}$); and

    \item the decision-bit associated with $s$ (denoted by $s.\mathsf{decision} \in \{\bot, 0, 1 \}$).
\end{compactitem} 
Intuitively, a state $s \in \mathcal{S}$, where (1) $s.\mathsf{process} = p_i$, (2) $s.\mathsf{round} = k$, (3) $s.\mathsf{proposal} = b$, and (4) $s.\mathsf{decision} = b'$, denotes the state of process $p_i$ at the start of round $k$ with $p_i$'s proposal being $b$ and $p_i$'s decision being $b'$ (if $b' = \bot$, $p_i$ has not yet decided by the start of round $k$).

For each process $p_i$, there are two \emph{initial states} $\pstatezero{i} \in \mathcal{S}$ and $\pstateone{i} \in \mathcal{S}$ associated with $p_i$ such that (1) $\pstatezero{i}.\mathsf{process} = \pstateone{i}.\mathsf{process} = p_i$, (2) $\pstatezero{i}.\mathsf{proposal} = 0$, and (3) $\pstateone{i}.\mathsf{proposal} = 1$.
Each process $p_i$ starts round $1$ in state $\pstatezero{i}$ or state $\pstateone{i}$.

\subsubsection{State-Transition Function.}
Algorithm $\mathcal{A}$ maps (1) the state of a process at the start of a round, and (2) messages the process received in the round into (a) a new state of the process at the start of the next round, and (b) messages the process sends in the next round.
Formally, given (1) a state $s \in \mathcal{S}$, and (2) a set of messages $M^R \subsetneq \mathcal{M}$ such that, for every message $m \in M^R$, $m.\mathsf{receiver} = s.\mathsf{process}$ and $m.\mathsf{round} = s.\mathsf{round}$, $\mathcal{A}(s, M^R) = (s', M^S)$, where 
\begin{compactitem}
    \item $s' \in \mathcal{S}$ is a state such that:
    \begin{compactitem}
        \item $s'.\mathsf{process} = s.\mathsf{process}$, 

        \item $s'.\mathsf{round} = s.\mathsf{round} + 1$, 

        \item $s'.\mathsf{proposal} = s.\mathsf{proposal}$,

        \item if $s.\mathsf{decision} \neq \bot$, $s'.\mathsf{decision} = s.\mathsf{decision}$; and
    \end{compactitem}

    \item $M^S \subsetneq \mathcal{M}$ is a set of messages such that, for every message $m \in M^S$, the following holds:
    \begin{compactitem}
        \item $m.\mathsf{sender} = s'.\mathsf{process} = s.\mathsf{process}$,

        \item $m.\mathsf{round} = s'.\mathsf{round} = s.\mathsf{round} + 1$,

        \item $m.\mathsf{receiver} \neq m.\mathsf{sender}$,

        \item there is no message $m' \in M^S$ such that (1) $m' \neq m$, and (2) $m.\mathsf{receiver} = m'.\mathsf{receiver}$.
    \end{compactitem}
\end{compactitem}
The messages each process $p_i$ sends in the first round depend solely on $p_i$'s initial state:
\begin{compactitem}
    \item If $p_i$'s state at the start of the first round is $\pstatezero{i}$, then $\mathcal{M}_i^0$ denotes the messages $p_i$ sends in the first round.
    For every message $m \in \mathcal{M}_i^0$, the following holds: (1) $m.\mathsf{sender} = p_i$, (2) $m.\mathsf{round} = 1$, (3) $m.\mathsf{receiver} \neq p_i$, and (4) there is no message $m' \in \mathcal{M}_i^0$ such that (a) $m' \neq m$, and (b) $m'.\mathsf{receiver} = m.\mathsf{receiver}$.

    \item If $p_i$'s state at the start of the first round is $\pstateone{i}$, then $\mathcal{M}_i^1$ denotes the messages $p_i$ sends in the first round. 
    For every message $m \in \mathcal{M}_i^1$, the following holds: (1) $m.\mathsf{sender} = p_i$, (2) $m.\mathsf{round} = 1$, (3) $m.\mathsf{receiver} \neq p_i$, and (4) there is no message $m' \in \mathcal{M}_i^1$ such that (a) $m' \neq m$, and (b) $m'.\mathsf{receiver} = m.\mathsf{receiver}$.
\end{compactitem}

\subsubsection{Fragments.}
A tuple $\mathcal{FR} = \Bigl( s, M^{S}, M^{\mathit{SO}}, M^R, M^{\mathit{RO}} \Bigl)$, where $s \in \mathcal{S}$ and $M^S, M^{\mathit{SO}}, M^R, M^{\mathit{RO}} \subsetneq \mathcal{M}$, is a \emph{$k$-round fragment}, for some $k \in \mathbb{N} \cup \{+ \infty\}$, of a process $p_i$ if and only if:
\begin{compactenum}
    \item $s.\mathsf{process} = p_i$; and

    \item $s.\mathsf{round} = k$; and

    \item for every message $m \in M^S \cup M^{\mathit{SO}} \cup M^R \cup M^{\mathit{RO}}$, $m.\mathsf{round} = k$; and

    \item $M^S \cap M^{\mathit{SO}} = \emptyset$; and

    \item $M^R \cap M^{\mathit{RO}} = \emptyset$; and

    \item for every message $m \in M^S \cup M^{\mathit{SO}}$, $m.\mathsf{sender} = p_i$; and

    \item for every message $m \in M^R \cup M^{\mathit{RO}}$, $m.\mathsf{receiver} = p_i$; and

    \item there is no message $m \in M^S \cup M^{\mathit{SO}} \cup M^R \cup M^{\mathit{RO}}$ such that $m.\mathsf{sender} = m.\mathsf{receiver} = p_i$; and

    \item there are no two messages $m, m' \in M^S \cup M^{\mathit{SO}}$ such that $m.\mathsf{receiver} = m'.\mathsf{receiver}$; and

    \item there are no two messages $m, m' \in M^R \cup M^{\mathit{RO}}$ such that $m.\mathsf{sender} = m'.\mathsf{sender}$.
\end{compactenum}
Intuitively, a $k$-round fragment of a process describes what happens at a process from the perspective of an omniscient external observer in the $k$-th round.

\paragraph{Intermediate results on fragments.}
We now present a few simple results.

\begin{lemma} \label{lemma:fragment_receive_omitted}
Consider any $k$-round ($k \in \mathbb{N} \cup \{+ \infty\}$) fragment $\mathcal{FR} = \Bigl( s, M^{S}, M^{\mathit{SO}}, M^R, M^{\mathit{RO}} \Bigl)$ of any process $p_i$, and any tuple $\mathcal{FR}' = \Bigl( s, M^{S}, M^{\mathit{SO}}, M^R, X \subsetneq \mathcal{M} \Bigl)$ for which the following holds:
\begin{compactenum}[(i)]
    \item for every message $m \in X$, $m.\mathsf{round} = k$; and 
    
    \item $M^{R} \cap  X = \emptyset$; and

    \item for every $m \in X$, $m.\mathsf{receiver} = p_i$; and

    \item there is no message $m \in X$, $m.\mathsf{sender} = m.\mathsf{receiver} = p_i$; and

    \item there are no two messages $m, m' \in M^{R} \cup X$ such that $m.\mathsf{sender} = m'.\mathsf{sender}$.
\end{compactenum}
Then, $\mathcal{FR}'$ is a $k$-round fragment of $p_i$.
\end{lemma}
\begin{proof}
To prove that $\mathcal{FR}'$ is a $k$-round fragment of $p_i$, we prove that all ten conditions hold for $\mathcal{FR}'$.
By the statement of the lemma, $\mathcal{FR}$ is a $k$-round fragment of $p_i$.
Conditions (1), (2), (4), (6), and (9) hold for $\mathcal{FR}'$ as the first four elements of the tuple $\mathcal{FR}'$ are identical to the first four elements of $\mathcal{FR}$.
Conditions (3), (5), (7), (8), and (10) hold due to conditions (i), (ii), (iii), (iv), and (v), respectively.
\end{proof}

\begin{lemma} \label{lemma:fragment_send_transfer}
Consider any $k$-round ($k \in \mathbb{N} \cup \{+\infty\}$) fragment $\mathcal{FR} = \Bigl( s, M^{S}, M^{\mathit{SO}}, M^R, M^{\mathit{RO}} \Bigl)$ of any process $p_i$, and any tuple $\mathcal{FR}' = \Bigl( s, X \subsetneq \mathcal{M}, Y \subsetneq \mathcal{M}, M^R, M^{\mathit{RO}}\Bigl)$ for which the following holds:
\begin{compactenum}[(i)]
    \item for every message $m \in X \cup Y$, $m.\mathsf{round} = k$; and 
    
    \item $X \cap  Y = \emptyset$; and

    \item for every message $m \in X \cup Y$, $m.\mathsf{sender} = p_i$; and

    \item there is no message $m \in X \cup Y$, $m.\mathsf{sender} = m.\mathsf{receiver} = p_i$; and

    \item there are no two messages $m, m' \in X \cup Y$ such that $m.\mathsf{receiver} = m'.\mathsf{receiver}$.
\end{compactenum}
Then, $\mathcal{FR}'$ is a $k$-round fragment of $p_i$.
\end{lemma}
\begin{proof}
Due to the statement of the lemma, $\mathcal{FR}$ is a $k$-round fragment of $p_i$. 
Therefore, conditions (1), (2), (5), (7), and (10) hold directly for $\mathcal{FR}'$ as the first, fourth, and fifth elements of $\mathcal{FR}'$ are identical to the first, fourth and fifth elements of $\mathcal{FR}$.
Conditions (3), (4), (6), (8), and (9) hold due to conditions (i), (ii), (iii), (iv), and (v), respectively.
\end{proof}

\subsubsection{Behaviors.}
In this subsection, we define behaviors of processes.
A tuple $\mathcal{B} = \Bigl \langle \mathcal{FR}^1 = \Bigl( s^1, M^{S(1)}, M^{\mathit{SO}(1)}, M^{\mathit{R}(1)}, M^{\mathit{RO}(1)} \Bigl), ..., \mathcal{FR}^k = \Bigl( s^k, M^{S(k)}, M^{\mathit{SO}(k)}, M^{\mathit{R}(k)}, M^{\mathit{RO}(k)} \Bigl) \Bigl \rangle$ is a \emph{$k$-round behavior}, for some $k \in \mathbb{N} \cup \{+ \infty\}$, of a process $p_i$ if and only if:
\begin{compactenum}
    \item for every $j \in [1, k]$, $\mathcal{FR}^j$ is a $j$-round fragment of $p_i$; and

    \item $s^1 = \pstatezero{i}$ or $s^1 = \pstateone{i}$; and

    \item if $s^1 = \pstatezero{i}$, then $M^{S(1)} \cup M^{\mathit{SO}(1)} = \mathcal{M}_i^0$; and

    \item if $s^1 = \pstateone{i}$, then $M^{S(1)} \cup M^{\mathit{SO}(1)} = \mathcal{M}_i^1$; and

    \item $s^1.\mathsf{proposal} = s^2.\mathsf{proposal} = ... = s^k.\mathsf{proposal}$; and

    \item if there exists $j \in [1, k]$ such that $s^j.\mathsf{decision} \neq \bot$, then there exists $j^* \in [1, j]$ such that (1) for every $j' \in [1, j^* - 1]$, $s^{j'}.\mathsf{decision} = \bot$, and (2) $s^{j^*}.\mathsf{decision} = s^{j^* + 1}.\mathsf{decision} = ... = s^k.\mathsf{decision}$; and
    
    \item for every $j \in [1, k - 1]$, $\mathcal{A}(s^j, M^{R(j)}) = (s^{j + 1}, M^{S(j+1)} \cup M^{\mathit{SO}(j+1)})$. \label{item:deterministic_transition}
\end{compactenum}
If $k = +\infty$, we say that $\mathcal{B}$ is an \emph{infinite behavior} of $p_i$.
Intuitively, a process's behavior describes the states and sets of sent and received messages (including those that are omitted) of that process.

\paragraph{Intermediate results on behaviors.}
We first introduce a few functions concerned with behaviors (see the \hyperref[algorithm:behavior_functions]{\emph{Functions}} table) before proving two intermediate results (lemmas \ref{lemma:behavior_receive_omitted} and \ref{lemma:behavior_send_transfer}).

\begin{algorithm}
\renewcommand{\thealgorithm}{}
\floatname{algorithm}{}
\caption{Functions defined on the $k$-round behavior $\mathcal{B}$ defined above}
\begin{algorithmic} [1]
\label{algorithm:behavior_functions}
\footnotesize
\State \textbf{function} $\fstate{\mathcal{B}}{j \in [1, k]}$:
\State \hskip2em \textbf{return} $s^j$ \BlueComment{returns the state at the start of round $j$}

\medskip
\State \textbf{function} $\fmsgs{\mathcal{B}}{j \in [1, k]}$:
\State \hskip2em \textbf{return} $M^{S(j)}$ \BlueComment{returns the messages (successfully) sent in round $j$}

\medskip
\State \textbf{function} $\fmsgso{\mathcal{B}}{j \in [1, k]}$:
\State \hskip2em \textbf{return} $M^{\mathit{SO}(j)}$ \BlueComment{returns the messages send-omitted in round $j$}

\medskip
\State \textbf{function} $\fmsgr{\mathcal{B}}{j \in [1, k]}$:
\State \hskip2em \textbf{return} $M^{R(j)}$ \BlueComment{returns the messages received in round $j$}

\medskip
\State \textbf{function} $\fmsgro{\mathcal{B}}{j \in [1, k]}$:
\State \hskip2em \textbf{return} $M^{\mathit{RO}(j)}$ \BlueComment{returns the messages receive-omitted in round $j$}

\medskip
\State \textbf{function} $\mathsf{all\_sent}(\mathcal{B})$:
\State \hskip2em \textbf{return} $\bigcup\limits_{j \in [1, k]} M^{\mathit{S}(j)}$ \BlueComment{returns all (successfully) sent messages}

\medskip
\State \textbf{function} $\mathsf{all\_send\_omitted}(\mathcal{B})$:
\State \hskip2em \textbf{return} $\bigcup\limits_{j \in [1, k]} M^{\mathit{SO}(j)}$ \BlueComment{returns all send-omitted messages}

\medskip
\State \textbf{function} $\mathsf{all\_receive\_omitted}(\mathcal{B})$:
\State \hskip2em \textbf{return} $\bigcup\limits_{j \in [1, k]} M^{\mathit{RO}(j)}$ \BlueComment{returns all receive-omitted messages}
\end{algorithmic}
\end{algorithm}

\begin{lemma} \label{lemma:behavior_receive_omitted}
Consider any $k$-round ($k \in \mathbb{N} \cup \{+ \infty\}$) behavior $\mathcal{B} = \Bigl \langle \mathcal{F}^1, \cdots, \mathcal{F}^k \Bigl \rangle$ of any process $p_i$, and any tuple $\mathcal{B}' = \Bigl \langle \mathcal{FR}^{1}, \cdots, \mathcal{FR}^{k}\Bigl \rangle$.
For every $j \in [1, k]$, $\mathcal{F}^j = \Bigl( s^j, M^{S(j)}, M^{\mathit{SO}(j)}, M^{\mathit{R}(j)}, M^{\mathit{RO}(j)} \Bigl)$.
Moreover, for every $j \in [1, k]$, $\mathcal{FR}^{j} = \Bigl( s^j, M^{S(j)}, M^{\mathit{SO}(j)}, M^{R(j)}, X^j \subsetneq \mathcal{M} \Bigl)$ and the following holds:
\begin{compactenum}[(i)]
    \item for every message $m \in X^j$, $m.\mathsf{round} = j$; and 

    \item $M^{R(j)} \cap  X^j = \emptyset$; and

    \item for every message $m \in X^j$, $m.\mathsf{receiver} = p_i$; and

    \item there is no message $m \in X^j$, $m.\mathsf{sender} = m.\mathsf{receiver} = p_i$; and

    \item there are no two messages $m, m' \in M^{R(j)} \cup X^j$ such that $m.\mathsf{sender} = m'.\mathsf{sender}$.
\end{compactenum}
Then, $\mathcal{B}'$ is a $k$-round behavior of $p_i$.
\end{lemma}

\begin{proof}
    Since $\mathcal{B}$ is a behavior of $p_i$, $\mathcal{F}^1, \cdots, \mathcal{F}^k$ are fragments of $p_i$.
    Thus, for every $j \in [1, k]$, $\mathcal{FR}^j$ is a $j$-round fragment of $p_i$ due to conditions (i) to (v) and \Cref{lemma:fragment_receive_omitted}, which implies that condition (1) holds for $\mathcal{B}'$.
    Conditions (3) and (4) hold for $\mathcal{B}'$ as, for every $j \in [1, k]$, the second and third elements of $\mathcal{FR}^j$ are identical to the second and third elements of $\mathcal{F}^j$.
    Similarly, conditions (2), (5) and (6) hold for $\mathcal{B}'$ as, for every $j \in [1, k]$, the first element of $\mathcal{FR}^j$ is identical to the state from $\mathcal{F}^j$.
    Finally, condition (7) holds for $\mathcal{B}'$: first, for every $j \in [1, k]$, the first four elements of $\mathcal{FR}^j$ are identical to the first four elements of $\mathcal{F}^j$; second, condition (7) holds for $\mathcal{B}$.
\end{proof}

\begin{lemma} \label{lemma:behavior_send_transfer}
Consider any $k$-round ($k \in \mathbb{N} \cup \{+ \infty\}$) behavior $\mathcal{B} = \Bigl \langle \mathcal{F}^1, \cdots, \mathcal{F}^k \Bigl \rangle$ of any process $p_i$ and a tuple $\mathcal{B}' = \Bigl \langle \mathcal{FR}^{1}, \cdots, \mathcal{FR}^{k}\Bigl \rangle$.
For every $j \in [1, k]$, $\mathcal{F}^j = \Bigl( s^j, M^{S(j)}, M^{\mathit{SO}(j)}, M^{\mathit{R}(j)}, M^{\mathit{RO}(j)} \Bigl)$.
Moreover, for every $j \in [1, k]$, $\mathcal{FR}^{j} = \Bigl( s^j, X^j \subsetneq \mathcal{M}, Y^j \subsetneq \mathcal{M}, M^{\mathit{R}(j)}, M^{\mathit{RO}(j)} \Bigl)$ such that (1) $X^j \cup Y^j = M^{S(j)} \cup M^{\mathit{SO}(j)}$, and (2) $X^j \cap Y^j = \emptyset$. 
Then, $\mathcal{B}'$ is a $k$-round behavior of $p_i$.
\end{lemma}
\begin{proof}
Since $\mathcal{B}$ is a behavior of $p_i$, $\mathcal{F}^1, \cdots, \mathcal{F}^k$ are fragments of $p_i$.
Thus, for every $j \in [1, k]$, $\mathcal{FR}^j$ is a $j$-round fragment of $p_i$ due to \Cref{lemma:fragment_send_transfer}, which implies that condition (1) holds for $\mathcal{B}'$.
Conditions (3) and (4) hold for $\mathcal{B}'$ as (1) both conditions hold for $\mathcal{B}$, and (2) for every $j \in [1, k]$, $X^j \cup Y^j = M^{S(j)} \cup M^{\mathit{SO}(j)}$.
Similarly, conditions (2), (5) and (6) hold for $\mathcal{B}'$ as, for every $j \in [1, k]$, the first element of $\mathcal{FR}^j$ is identical to the state from $\mathcal{F}^j$.
Finally, condition (7) holds for $\mathcal{B}'$: first, for every $j \in [1, k]$, $X^j \cup Y^j = M^{S(j)} \cup M^{\mathit{SO}(j)}$ and the first and the fourth elements of $\mathcal{FR}^j$ are identical to the first and the fourth elements of $\mathcal{F}^j$; second, condition (7) holds for $\mathcal{B}$.
\end{proof}

\subsubsection{Executions.} \label{subsubsection:executions}
A \emph{$k$-round execution} $\mathcal{E}$, for some $k \in \mathbb{N} \cup \{+ \infty\}$, is a tuple $\mathcal{E} = [\mathcal{F} \subsetneq \Pi, \mathcal{B}_1,  ..., \mathcal{B}_n]$ such that the following guarantees hold:
\begin{compactitem}
    \item \emph{Faulty processes:} $\mathcal{F}$ is a set of $|\mathcal{F}| \leq t$ processes.
    
    \item \emph{Composition:} For every $j \in [1, n]$, $\mathcal{B}_j$ is a $k$-round behavior of process $p_j$.

    \item \emph{Send-validity:} If there exists a message $m$, where $p_s = m.\mathsf{sender}$, $p_r = m.\mathsf{receiver}$ and $j = m.\mathsf{round}$, such that $m \in \mathsf{sent}(\mathcal{B}_s, j)$, then the following holds: $m \in \mathsf{received}(\mathcal{B}_r, j)$ or $m \in \mathsf{receive\_omitted}(\mathcal{B}_r, j)$.
    That is, if a message is (successfully) sent, the message is either received or receive-omitted in the same round.

    \item \emph{Receive-validity:} If there exists a message $m$, where $p_s = m.\mathsf{sender}$, $p_r = m.\mathsf{receiver}$ and $j = m.\mathsf{round}$, such that $m \in \mathsf{received}(\mathcal{B}_r, j) \cup \fmsgro{\mathcal{B}_r}{j}$, then $m \in \mathsf{sent}(\mathcal{B}_s, j)$.
    That is, if a message is received or receive-omitted, the message is (successfully) sent in the same round.

    \item \emph{Omission-validity:} If there exists a process $p_i$ and $j \in [1, k]$ such that (1) $\mathsf{send\_omitted}(\mathcal{B}_i, j) \neq \emptyset$, or (2) $\mathsf{receive\_omitted}(\mathcal{B}_i, j) \neq \emptyset$, then $p_i \in \mathcal{F}$.
    That is, if a process commits an omission fault, the process belongs to $\mathcal{F}$.
\end{compactitem}
If $k = +\infty$, we say that $\mathcal{E}$ is an \emph{infinite execution}.

\subsection{Preliminary Lemmas} \label{subsection:prelimariny_lemmas_appendix}

We start by defining the $\mathsf{swap\_omission}$ procedure (\Cref{algorithm:swap}).


\begin{algorithm}
\caption{Procedure $\mathsf{swap\_omission}$}
\label{algorithm:swap}
\begin{algorithmic} [1]
\footnotesize
\State \textbf{procedure} $\mathsf{swap\_omission}(\mathsf{Execution} \text{ } \mathcal{E} = [\mathcal{F}, \mathcal{B}_1, ..., \mathcal{B}_n], p_i \in \Pi)$:
\State \hskip2em \textbf{let} $M \gets \mathsf{all\_receive\_omitted}(\mathcal{B}_i)$ \BlueComment{$M$ contains all messages which are receive-omitted by $p_i$} \label{line:M_swap}
\State \hskip2em \textbf{let} $\mathcal{F}' \gets \emptyset$ \BlueComment{new set of faulty processes}
\State \hskip2em \textbf{for each} $p_z \in \Pi$:
\State \hskip4em $\textbf{let}$ $\mathcal{B}_z = \langle \mathcal{FR}_z^1, ..., \mathcal{FR}_z^k \rangle$, for some $k \in \mathbb{N} \cup \{+ \infty\}$
\State \hskip4em \textbf{for each} $j \in [1, k]$:
\State \hskip6em \textbf{let} $\mathit{sent}_z \gets \{m \in \mathcal{M} \,|\, m \in M \land m.\mathsf{round} = j \land m.\mathsf{sender} = p_z \}$ \BlueComment{messages from $M$ sent by $p_z$} \label{line:update_sent_z}
\State \hskip6em \textbf{let} $\mathcal{FR}_z^j = (s^j, M^{S(j)}, M^{\mathit{SO}(j)}, M^{R(j)}, M^{\mathit{RO}(j)})$ \BlueComment{old fragment}
\State \hskip6em \textbf{let} $\mathcal{FR}^j \gets (s^j, M^{S(j)} \setminus \mathit{sent}_z, M^{\mathit{SO}(j)} \cup \mathit{sent}_z, M^{R(j)}, M^{\mathit{RO}(j)} \setminus{M})$ \label{line:sent_swap} \BlueComment{new fragment}
\State \hskip6em \textbf{if} $(M^{\mathit{SO}(j)} \cup \mathit{sent}_z) \cup (M^{\mathit{RO}(j)} \setminus{M}) \neq \emptyset$: \label{line:check_swap} \BlueComment{check for an omission fault}
\State \hskip8em $\mathcal{F}' \gets \mathcal{F}' \cup \{p_z\}$ \BlueComment{$p_z$ is faulty}
\State \hskip4em \textbf{let} $\mathcal{B}_z' \gets \langle \mathcal{FR}^1, ..., \mathcal{FR}^k \rangle$
\State \hskip2em \textbf{return} $[\mathcal{F}', \mathcal{B}_1', ..., \mathcal{B}_n']$
\end{algorithmic}
\end{algorithm}

Intuitively, $\mathsf{swap\_omission}(\mathcal{E}, p_i)$, for some execution $\mathcal{E}$ and process $p_i$, constructs an execution $\mathcal{E}'$ in which receive-omission faults of process $p_i$ are ``swapped'' for send-omission faults of other processes.
The following lemma proves that, if some preconditions are true, $\mathcal{E}'$ is indeed an execution and it satisfies certain properties.

\begin{lemma}\label{lemma:switch_omission}
Let $\mathcal{E} = [\mathcal{F}, \mathcal{B}_1, ..., \mathcal{B}_n]$ be any $k$-round ($k \in \mathbb{N} \cup \{+ \infty \}$) execution.
Moreover, let $[\mathcal{F}', \mathcal{B}_1', ..., \mathcal{B}_n'] \gets \mathsf{swap\_omission}(\mathcal{E}, p_i)$, for some process $p_i$.
Let the following hold:
\begin{compactitem}
    \item $|\mathcal{F}'| \leq t$; and

    \item $\mathsf{all\_send\_omitted}(\mathcal{B}_i) = \emptyset$; and

    \item there exists a process $p_h \in \Pi \setminus{\mathcal{F}}$ such that $p_h \neq p_i$ and $\mathsf{all\_sent}(\mathcal{B}_h) \cap \mathsf{all\_receive\_omitted}(\mathcal{B}_i) = \emptyset$.
\end{compactitem}
Then, (1) $[\mathcal{F}', \mathcal{B}_1', ..., \mathcal{B}_n']$ is a $k$-round execution, (2) $\mathcal{E}$ and $[\mathcal{F}', \mathcal{B}_1', ..., \mathcal{B}_n']$ are indistinguishable to every process $p_z \in \Pi$,  (3) $p_i \notin \mathcal{F}'$, and (4) $p_h \notin \mathcal{F}'$. 
\end{lemma}

\begin{proof}
To prove the lemma, we first prove that all guarantees that an execution needs to satisfy (see \Cref{subsubsection:executions}) are indeed satisfied for the tuple $[\mathcal{F}', \mathcal{B}_1', ..., \mathcal{B}_n']$.
\begin{compactitem}
    \item \emph{Faulty processes:} Follows from the precondition of the lemma.

    \item \emph{Composition:} As $\mathcal{E}$ is a $k$-round execution, $\mathcal{B}_i$ is a $k$-round behavior of every process $p_i$.
    Therefore, for every process $p_i$, $\mathcal{B}'_i$ is a $k$-round behavior of $p_i$ due to \cref{lemma:behavior_receive_omitted,lemma:behavior_send_transfer}.
    
    \item \emph{Send-validity:}
    Consider any message $m$, where $p_s = m.\mathsf{sender}$, $p_r = m.\mathsf{receiver}$ and $j = m.\mathsf{round}$, such that $m$ is sent in $\mathcal{B}'_s$.
    Note that $m \in \fmsgs{\mathcal{B}_s}{j}$ (as no new sent messages are added to $\mathcal{B}_s'$ at line~\ref{line:sent_swap}).
    Therefore, $m \in \mathsf{sent}(\mathcal{B}_s', j)$ and $m \in \fmsgr{\mathcal{B}_r}{j} \cup \fmsgro{\mathcal{B}_r}{j}$ (due to the send-validity property of $\mathcal{E}$). 
    As $m$ is sent in $\mathcal{B}'_s$ (i.e., $m \in M^{S(j)}$ at process $p_s$; line~\ref{line:sent_swap}), $m \notin M$. 
    Thus, $m$ is not excluded from $M^{\mathit{RO}(j)}$ at  process $p_r$ (line~\ref{line:sent_swap}), which implies $m \in M^{R(j)} \cup (M^{\mathit{RO}(j)} \setminus{M})$ at process $p_r$.
    Thus, send-validity holds.

    \item \emph{Receive-validity:}
    Consider any message $m$, where $p_s = m.\mathsf{sender}$, $p_r = m.\mathsf{receiver}$ and $j = m.\mathsf{round}$, such that $m$ is received or receive-omitted in $\mathcal{B}_r'$.
    As $m$ is received or receive-omitted in $\mathcal{B}_r'$, $m$ is received or receive-omitted in $\mathcal{B}_r$ (as no new received or receive-omitted messages are added to $\mathcal{B}_r'$ at line~\ref{line:sent_swap}).
    Moreover, $m \in \mathsf{received}(\mathcal{B}_r, j) \cup \mathsf{receive\_omitted}(\mathcal{B}_r, j)$ (as $\mathcal{B}_r$ is $k$-round behavior of $p_r$), which then implies that $m \in \mathsf{sent}(\mathcal{B}_s, j)$ (as $\mathcal{E}$ satisfies receive-validity).
    Furthermore, $m \notin M$; otherwise, $m$ would not be received nor receive-omitted in $\mathcal{B}_r'$.
    Therefore, $m$ is not excluded from $M^{S(j)}$ at process $p_s$ (line~\ref{line:sent_swap}), which proves that receive-validity is satisfied.
    
    \item \emph{Omission-validity:} Follows directly from the check at line~\ref{line:check_swap}.
\end{compactitem}
As all guarantees are satisfied, $[\mathcal{F}', \mathcal{B}_1', ..., \mathcal{B}_n']$ is a $k$-round execution, which proves the first statement of the lemma.

Second, we prove the indistinguishability statement for every process $p_z \in \Pi$.
The $\mathsf{swap\_omission}$ procedure (\Cref{algorithm:swap}) ensures that $\mathsf{received}(\mathcal{B}_z', j) = \mathsf{received}(\mathcal{B}_z, j)$ (line~\ref{line:sent_swap}), for every round $j \in [1, k]$.
Moreover, for every round $j \in [1, k]$, $\mathsf{state}(\mathcal{B}_z', j) = \mathsf{state}(\mathcal{B}_z, j)$ and $\mathsf{sent}(\mathcal{B}_z', j) \cup \mathsf{send\_omitted}(\mathcal{B}_z', j) = \mathsf{sent}(\mathcal{B}_z, j) \cup \mathsf{send\_omitted}(\mathcal{B}_z, j)$ (line~\ref{line:sent_swap}).
Therefore, the second statement of the lemma holds.

Third, we prove that $p_i \notin \mathcal{F}'$.
As no process sends messages to itself, $\mathit{sent}_i = \emptyset$ (line~\ref{line:update_sent_z}) in every round $j \in [1, k]$.
Hence, $\mathsf{all\_send\_omitted}(\mathcal{B}'_i) = \emptyset$ (line~\ref{line:sent_swap}).
Moreover, $M = \mathsf{all\_receive\_omitted}(\mathcal{B}_i) = \bigcup\limits_{j \in [1, k]} M^{\mathit{RO}(j)}$ (line~\ref{line:M_swap}).
Therefore, $\mathsf{all\_receive\_omitted}(\mathcal{B}'_i) = \emptyset$, which implies that the third statement of the lemma holds.

Finally, we prove that $p_h \notin \mathcal{F}'$.
As $p_h \notin \mathcal{F}$, $\mathsf{all\_send\_omitted}(\mathcal{B}_h) = \mathsf{all\_receive\_omitted}(\mathcal{B}_h) = \emptyset$.
As $\mathsf{all\_receive\_omitted}(\mathcal{B}_h) = \emptyset$, $\mathsf{all\_receive\_omitted}(\mathcal{B}_h') = \emptyset$ (line~\ref{line:sent_swap}).
Moreover, $\mathit{sent}_h = \emptyset$ (line~\ref{line:update_sent_z}) in every round $j \in [1, k]$.
Therefore, $\mathsf{all\_send\_omitted}(\mathcal{B}_h') = \emptyset$ (line~\ref{line:sent_swap}).
Hence, $p_h \notin \mathcal{F}'$, which concludes the proof of the lemma.
\end{proof}

\Cref{algorithm:merge} defines the $\mathsf{merge}$ procedure which constructs a new execution from two mergeable ones; recall that mergeable executions are defined by \Cref{definition:mergeable}.
The following lemma proves that the result of the $\mathsf{merge}$ procedure (\Cref{algorithm:merge}) is an execution that is indistinguishable from the original one and satisfies some important properties.

\begin{algorithm}
\caption{Procedure $\mathsf{merge}$}
\label{algorithm:merge}
\begin{algorithmic} [1]
\footnotesize
\State \textbf{procedure} $\mathsf{merge}(\mathsf{Execution} \text{ } \mathcal{E}_0^{B(k_1)} = [B, \mathcal{B}_1, ..., \mathcal{B}_n], \mathsf{Execution} \text{ } \mathcal{E}_b^{C(k_2)} = [C, \mathcal{B}_1', ..., \mathcal{B}'_n])$:
\State \hskip2em \textbf{assert} ($\mathcal{E}_0^{B(k_1)}$ and $\mathcal{E}_b^{C(k_2)}$ are mergeable executions)
\State \hskip2em \textbf{let} $\mathit{sent} \gets \bigcup\limits_{p_i \in A \cup B} \mathcal{M}_i^0 \cup \bigcup\limits_{p_i \in C} \mathcal{M}_i^b$ \BlueComment{messages sent in the first round} 

\State \hskip2em \textbf{let} $s_i \gets 0_i$, for every process $p_i \in A \cup B$ \BlueComment{the initial state of processes in $A \cup B$}
\State \hskip2em \textbf{let} $\mathit{sent}_i \gets \mathcal{M}_i^0$, for every process $p_i \in A \cup B$ \BlueComment{the initial messages sent by processes in $A \cup B$}

\State \hskip2em \textbf{let} $s_i \gets b_i$, for every process $p_i \in C$ \BlueComment{the initial state of processes in $C$}
\State \hskip2em \textbf{let} $\mathit{sent}_i \gets \mathcal{M}_i^b$, for every process $p_i \in C$ \BlueComment{the initial messages sent by processes in $C$}

\State \hskip2em \textbf{for each} $j \geq 1$: \label{line:for_loop_merge}
\State \hskip4em \textbf{for each} $p_i \in \Pi$:
\State \hskip6em \textbf{let} $\mathit{to}_i \gets \{m \,|\, m \in \mathit{sent} \land m.\mathsf{receiver} = p_i\}$ \label{line:to_r} \BlueComment{messages sent in the round $j$ to $p_i$}
\State \hskip6em \textbf{let} $\mathit{received}_i \gets \emptyset$

\State \hskip6em \textbf{if} $p_i \in A$:
\State \hskip8em \textbf{let} $\mathcal{FR}_i^j = (s_i, \mathit{sent}_i, \emptyset, \mathit{to}_i, \emptyset)$ \label{line:fragment_a_merge}
\State \hskip8em $\mathit{received}_i \gets \mathit{to}_i$
\State \hskip6em \textbf{else:}
\State \hskip8em \textbf{if} $p_i \in B$: $\mathit{received}_i \gets \mathsf{received}(\mathcal{B}_i, j)$ \BlueComment{receive messages from $\mathcal{B}_i$} \label{line:received_b}
\State \hskip8em \textbf{else:} $\mathit{received}_i \gets \mathsf{received}(\mathcal{B}_i', j)$ \BlueComment{receive messages from $\mathcal{B}'_i$} \label{line:received_c}
\State \hskip8em \textbf{let} $\mathcal{FR}_i^j = (s_i, \mathit{sent}_i, \emptyset, \mathit{received}_i, \mathit{to}_i \setminus{\mathit{received}_i})$ \label{line:fragment_bc_merge}

\State \hskip6em $(s_i, \mathit{sent}_i) \gets \mathcal{A}(s_i, \mathit{received}_i)$ \BlueComment{compute new state and newly sent messages} \label{line:merge_state_machine}
\State \hskip4em $\mathit{sent} \gets \bigcup\limits_{p_i \in \Pi} \mathit{sent}_i$ \label{line:update_sent_2} \BlueComment{update sent messages}

\State \hskip2em \textbf{for each} $p_i \in \Pi$:
\State \hskip4em \textbf{let} $\mathcal{B}_i^* = \langle \mathcal{FR}_i^1, \mathcal{FR}_i^2, ... \rangle
$
\State \hskip2em \textbf{return} $[B \cup C, \mathcal{B}_1^*, ..., \mathcal{B}_n^*]$
\end{algorithmic}
\end{algorithm}

\begin{lemma} \label{lemma:merge_guarantees}
Let executions $\mathcal{E}_0^{B(k_1)}$ ($k_1 \in \mathbb{N}$) and $\mathcal{E}_b^{C(k_2)}$ ($b \in \{0, 1\}$, $k_2 \in \mathbb{N}$) be mergeable.
Then, (1) $\mathcal{E}^* = \mathsf{merge}(\mathcal{E}_0^{B(k_1)}, \mathcal{E}_b^{C(k_2)})$ is an infinite execution, (2) $\mathcal{E}^*$ is indistinguishable from $\mathcal{E}_0^{B(k_1)}$ (resp., $\mathcal{E}_b^{C(k_2)}$) to every process $p_B \in B$ (resp., $p_C \in C$), and (3) group $B$ (resp., $C$) is isolated from round $k_1$ (resp., $k_2$) in $\mathcal{E}^{*}$.
\end{lemma}
\begin{proof}
Let $\mathcal{E}^* = [B \cup C, \mathcal{B}_1^*, ..., \mathcal{B}_n^*]$.
Let $\mathcal{E}_0^{B(k_1)} = [B, \mathcal{B}_1, ..., \mathcal{B}_n]$ and $\mathcal{E}_b^{C(k_2)} = [C, \mathcal{B}_1', ..., \mathcal{B}_n']$. 
To prove the lemma, we first prove that all guarantees from \Cref{subsubsection:executions} are satisfied by $\mathcal{E}^*$:
\begin{compactitem}
    \item \emph{Faulty processes:} As $\mathcal{F}^* = B \cup C$, $|\mathcal{F}^*| = \frac{t}{4} + \frac{t}{4} \leq t$.

    \item \emph{Composition:} For each process $p_i \in \Pi$, we construct $\mathcal{B}_i^*$ by following the valid transitions of the algorithm $\mathcal{A}$ (line~\ref{line:merge_state_machine}).
    Thus, for each process $p_i \in \Pi$, $\mathcal{B}_i^*$ is an infinite behavior of $p_i$.

    \item \emph{Send-validity:} Consider any message $m$, where $p_s = m.\mathsf{sender}$, $p_r = m.\mathsf{receiver}$ and $j = m.\mathsf{round}$, such that $m$ is sent $\mathcal{B}_s^*$.
    As $m.\mathsf{round} = j$, $m \in \mathit{sent}_s$ in the $j$-th iteration of the for loop at line~\ref{line:for_loop_merge}.
    Therefore, $m \in \mathit{to}_r$ (line~\ref{line:to_r}) in the $j$-th iteration of the for loop at line~\ref{line:for_loop_merge}.
    Hence, $m \in \mathsf{received}(\mathcal{B}_r^*, j) \cup \mathsf{receive\_omitted}(\mathcal{B}_r^*, j)$ (line~\ref{line:fragment_a_merge} or line~\ref{line:fragment_bc_merge}).

    \item \emph{Receive-validity:} Consider any message $m$, where $p_s = m.\mathsf{sender}$, $p_r = m.\mathsf{receiver}$ and $j = m.\mathsf{round}$, such that $m$ is received or receive-omitted in $\mathcal{B}^*_r$.
    We distinguish three cases:
    \begin{compactitem}
        \item Let $p_r \in A$.
        In this case, $m \in \mathit{to}_r$ (line~\ref{line:to_r}) in the $j$-th iteration of the for loop at line~\ref{line:for_loop_merge}.
        Therefore, $m \in \mathit{sent}_s$ in the $j$-th iteration of the for loop at line~\ref{line:for_loop_merge}.
        Hence, $m \in \mathsf{sent}(\mathcal{B}_s^*, j)$ (line~\ref{line:fragment_a_merge} or line~\ref{line:fragment_bc_merge}).

        \item Let $p_r \in B$.
        We further distinguish two scenarios:
        \begin{compactitem}
            \item Let $m$ be received in $\mathcal{B}_r^*$.
            In this case, $m \in \mathsf{received}(\mathcal{B}_r, j)$ (line~\ref{line:received_b}).
            As $\mathcal{E}_0^{B(k_1)}$ satisfies receive-validity, $m \in \mathsf{sent}(\mathcal{B}_s, j)$.
            If $j < k_1$, then $m \in \mathsf{sent}(\mathcal{B}_s^*, j)$ as, until (and excluding) round $k_1$, all processes send the same messages as in $\mathcal{E}_0^{B(k_1)}$.
            Otherwise, $s \in B$ (as $B$ is isolated from round $k_1$ in $\mathcal{E}_0^{B(k_1)}$), which implies that $m \in \mathsf{sent}(\mathcal{B}_s^*, j)$.

            \item Let $m$ be receive-omitted in $\mathcal{B}_r^*$.
            In this case, $m \in to_r$ (line~\ref{line:to_r}) in the $j$-th iteration of the for loop at line~\ref{line:for_loop_merge}.
            Hence, $\mathit{m} \in \mathit{sent}_s$ in the $j$-th iteration of the for loop at line~\ref{line:for_loop_merge}, which implies that $m \in \mathsf{sent}(\mathcal{B}_s^*, j)$.
        \end{compactitem}

        \item Let $p_r \in C$.
        There are two scenarios:
        \begin{compactitem}
            \item Let $m$ be received in $\mathcal{B}_r^*$.
            In this case, $m \in \mathsf{received}(\mathcal{B}_r', j)$ (line~\ref{line:received_c}).
            As $\mathcal{E}_b^{C(k_2)}$ satisfies receive-validity, $m \in \mathsf{sent}(\mathcal{B}_s', j)$.
            If $j < k_2$, then $m \in \mathsf{sent}(\mathcal{B}_s^*, j)$ as, until (and excluding) round $k_2$, all processes send the same messages as in $\mathcal{E}_b^{C(k_2)}$.
            Otherwise, $s \in C$ (as $C$ is isolated from round $k_2$ in $\mathcal{E}_b^{C(k_2)}$), which implies that $m \in \mathsf{sent}(\mathcal{B}_s^*, j)$.

            \item Let $m$ be receive-omitted in $\mathcal{B}_r^*$.
            In this case, $m \in to_r$ (line~\ref{line:to_r}) in the $j$-th iteration of the for loop at line~\ref{line:for_loop_merge}.
            Hence, $\mathit{m} \in \mathit{sent}_s$ in the $j$-th iteration of the for loop at line~\ref{line:for_loop_merge}, which implies that $m \in \mathsf{sent}(\mathcal{B}_s^*, j)$.
        \end{compactitem}
    \end{compactitem}

    \item \emph{Omission-validity:} Only processes in $B \cup C$ (potentially) receive-omit some messages.
    As $\mathcal{F}^* = B \cup C$, omission-validity is satisfied.
\end{compactitem}

Second, we prove the indistinguishability statement. Let $p_i \in B$. 
The $\mathsf{merge}$ procedure (\Cref{algorithm:merge}) ensures that $\mathsf{received}(\mathcal{B}_i^*, j) = \mathsf{received}(\mathcal{B}_i, j)$, for every round $j \geq 1$.
Moreover, for every round $j \geq 1$, $\mathsf{state}(\mathcal{B}_i^*, j) = \mathsf{state}(\mathcal{B}_i, j)$ and $\mathsf{sent}(\mathcal{B}_i^*, j) \cup \mathsf{send\_omitted}(\mathcal{B}_i^*, j) = \mathsf{sent}(\mathcal{B}_i, j) \cup \mathsf{send\_omitted}(\mathcal{B}_i, j)$.
The symmetric argument holds if $p_i \in C$.
Hence the indistinguishability statement holds.

Finally, every process $p_B \in B$ (resp., $p_C \in C$) exhibits the same behavior (except for potentially new receive-omitted messages) in $\mathcal{E}^*$ as in $\mathcal{E}_0^{B(k_1)}$ (resp., $\mathcal{E}_b^{C(k_2)}$).
Therefore, group $B$ (resp., $C$) is isolated from round $k_1$ (resp., $k_2$) in $\mathcal{E}^*$.
\end{proof}

\subsection{\Cref{lemma:lower_bound_helper,lemma:either_or}} 
\label{appendix:formal_lower_bound} 
\label{subsection:proofs_appendix}

First, we formally prove \Cref{lemma:lower_bound_helper}.

\genericisolation*
\begin{proof}
Let $\mathcal{E} = [\mathcal{F}, \mathcal{B}_1, ..., \mathcal{B}_n]$, where $\mathcal{F} = Y \cup Z$.
For every process $p_i \in Y$, we define $\mathcal{M}_{X \to p_i}$:
\begin{equation*}
    \mathcal{M}_{X \to p_i} = \{m \in \mathsf{all\_receive\_omitted}(\mathcal{B}_i) \,|\, m.\mathsf{sender} \in X\}.
\end{equation*}
For every set $Y'' \subseteq Y$, let $\mathcal{M}_{X \to Y''} = \bigcup\limits_{p \in Y''} \mathcal{M}_{X \to p}$.
As correct processes (i.e., processes from group $X$) send fewer than $\frac{t^2}{32}$ messages in $\mathcal{E}$, $|\mathcal{M}_{X \to Y}| < \frac{t^2}{32}$.
Therefore, there does not exist a set $Y^* \subseteq Y$ of $|Y^*| \geq \frac{|Y|}{2}$ processes such that, for every process $p_{Y^*} \in Y^*$, $|\mathcal{M}_{X \to p_{Y^*}}| \geq \frac{t}{2}$.
This implies that there exists a set $Y' \subseteq Y$ of $|Y'| > \frac{|Y|}{2}$ processes such that, for every process $p_{Y'} \in Y'$, $|\mathcal{M}_{X \to p_{Y'}}| < \frac{t}{2}$.

Fix any process $p_{Y'} \in Y'$.
By contradiction, suppose that $p_{Y'}$ does not decide $b_X$ in $\mathcal{E}$.
We define a set of processes $\mathcal{S}$:
\begin{equation*}
    \mathcal{S} = \{ p_s \in \Pi \,|\,\exists m \in \mathsf{all\_receive\_omitted}(\mathcal{B}_{Y'}) : m.\mathsf{sender} = p_s \}.
\end{equation*}
Note that $|\mathcal{S} \cap X| < \frac{t}{2}$ (since $|\mathcal{M}_{X \to p_{Y'}}| < \frac{t}{2}$) and $\mathcal{S} \subsetneq X \cup Z$.
Let $\mathcal{E}' = [F', \mathcal{B}_1', ..., \mathcal{B}_n'] = \mathsf{swap\_omission}(\mathcal{E}, p_{Y'})$.
By \Cref{lemma:switch_omission}, $\mathcal{E}'$ is an execution that is indistinguishable from $\mathcal{E}$ to every process.
Moreover, \Cref{lemma:switch_omission} proves that $p_{Y'} \notin \mathcal{F}'$ (since $p_{Y'}$ does not commit any send omission fault) and that there exists a process $p_X \in X \setminus \mathcal{S}$ (since $X \setminus \mathcal{S} \neq \emptyset$) such that $p_X \notin \mathcal{F}'$.
Importantly, $p_X$ decides $b_X$ in $\mathcal{E}'$.
Therefore, $\mathcal{E}'$ violates either \emph{Termination} (if $p_{Y'}$ does not decide) or \emph{Agreement} (if $p_{Y'}$ decides $1 - b_X$), which concludes the proof.
\end{proof}

Lastly, we prove \Cref{lemma:either_or}.

\either*
\begin{proof}
For $\mathcal{A}$ to satisfy \emph{Termination} and \emph{Agreement}, all processes from group $A$ decide $b_1$ (resp., $b_2$) in $\mathcal{E}_0^{B(k_1)}$ (resp., $\mathcal{E}_b^{C(k_2)}$).
Given the partition $(A \cup C, B, \emptyset)$ of $\Pi$ and the execution $\mathcal{E}_0^{B(k_1)}$, \Cref{lemma:lower_bound_helper} proves that there exists a set $B' \subseteq B$ of more than $\frac{|B|}{2}$ processes such that every process $p_{B'} \in B'$ decides $b_1$ in $\mathcal{E}_0^{B(k_1)}$.
Similarly, given the partition $(A \cup B, C, \emptyset)$ of $\Pi$ and the execution $\mathcal{E}_b^{C(k_2)}$, \Cref{lemma:lower_bound_helper} proves that there exists a set $C' \subseteq C$ of more than $\frac{|C|}{2}$ processes such that every process $p_{C'} \in C'$ decides $b_2$ in $\mathcal{E}_b^{C(k_2)}$.

Let $\mathcal{E} = \mathsf{merge}(\mathcal{E}_0^{B(k_1)}, \mathcal{E}_b^{C(k_2)})$.
By \Cref{lemma:merge_guarantees}, $\mathcal{E}$ is an infinite execution such that group $B$ (resp., $C$) is isolated from round $k_1$ (resp., $k_2$) and no process $p_B \in B$ (resp., $p_C \in C$) distinguishes $\mathcal{E}$ from $\mathcal{E}_0^{B(k_1)}$ (resp., $\mathcal{E}_b^{C(k_2)}$).
Therefore, all processes from $B'$ (resp., $C'$) decide $b_1$ (resp., $b_2$) in $\mathcal{E}$.
Let $b_A$ be the decision of processes from group $A$ in $\mathcal{E}$; such a decision must exist as $\mathcal{A}$ satisfies \emph{Termination} and \emph{Agreement} and processes from group $A$ are correct in $\mathcal{E}$.
Given the partition $(A, B, C)$ of $\Pi$ and the newly constructed execution $\mathcal{E}$, \Cref{lemma:lower_bound_helper} proves that $b_1 = b_A$.
Similarly, given the partition $(A, C, B)$ of $\Pi$ and the execution $\mathcal{E}$, \Cref{lemma:lower_bound_helper} shows that $b_2 = b_A$.
As $b_1 = b_A$ and $b_A = b_2$, $b_1 = b_2$, which concludes the proof.
\end{proof}

%% file: appendix/canonical_containment.tex
\section{Proof of Lemma~\ref{lemma:canonical_containment}} \label{section:containment_proof}

This section provides a proof of \Cref{lemma:canonical_containment} (introduced in \Cref{subsection:reduction}).

\containment*
\begin{proof}
We prove the lemma by contradiction.
Hence, let $v' \in \mathcal{V}_O$ be decided in $\mathcal{E}$, and let $v' \notin \bigcap\limits_{c' \in \mathit{Cnt}(c)} \mathit{val}(c')$.
Therefore, there exists an input configuration $c' \in \mathcal{I}$ such that (1) $c \sqsupseteq c'$, and (2) $v' \notin \mathit{val}(c')$.
Observe that $\pi(c') \subsetneq \pi(c)$.
Let $p_i$ be any process such that $p_i \in \pi(c')$; note that such a process exists as $|\pi(c')| \geq n - t$ and $n - t > 0$.
If $\mathcal{E}$ is an infinite execution, then $\mathcal{E}_i \gets \mathcal{E}$; otherwise, let $\mathcal{E}_i$ be any infinite continuation of $\mathcal{E}$ such that $\mathit{Correct}_{\mathcal{A}}(\mathcal{E}_i) = \mathit{Correct}_{\mathcal{A}}(\mathcal{E})$.
Note that $c = \mathsf{input\_conf}(\mathcal{E}_i)$ and $p_i$ decides $v'$ in $\mathcal{E}_i$ (to satisfy \emph{Termination} and \emph{Agreement}).

Consider now another infinite execution $\mathcal{E}' \in \mathit{execs}(\mathcal{A})$ which is identical to $\mathcal{E}_i$, except that only processes in $\pi(c')$ are correct in $\mathcal{E}'$.
As $c \sqsupseteq c'$, each process which belongs to $\pi(c')$ has identical proposals in $c'$ and $c$, which implies that $c' = \mathsf{input\_conf}(\mathcal{E}')$.
Moreover, observe that (1) $p_i \in \mathit{Correct}_{\mathcal{A}}(\mathcal{E}')$ (as $p_i \in \pi(c')$), and (2) $p_i$ decides $v'$ in $\mathcal{E}'$ as $\mathcal{E}'$ and $\mathcal{E}_i$ are indistinguishable to $p_i$.
Thus, we reach a contradiction as a correct process (namely, $p_i$) decides value $v' \notin \mathit{val}(c')$ in execution $\mathcal{E}'$ which corresponds to $c'$, thus violating $\mathit{val}$.
Consequently, the lemma holds.
\end{proof}

%% file: appendix/reduction_proof.tex
\section{Proof of the Reduction (Algorithm~\ref{algorithm:reduction})} \label{section:reduction_proof}

In this section, we prove the correctness of the reduction (\Cref{algorithm:reduction}) from weak consensus to any solvable non-trivial agreement problem $\mathcal{P}$.
First, we prove that $v_0' \neq v_1'$ (see \Cref{table:notation_reduction}).
Recall that $v_0'$ (resp., $v_1'$) is decided in the fully correct infinite execution $\mathcal{E}_0$ (resp., $\mathcal{E}_1$) of $\mathcal{A}$ such that $\mathsf{input\_conf}(\mathcal{E}_0) = c_0$ (resp., $\mathsf{input\_conf}(\mathcal{E}_1) = c_1$), where $\mathcal{A}$ is the fixed algorithm which solves $\mathcal{P}$.

\begin{lemma} \label{lemma:reduction_helper}
$v_0' \neq v_1'$.
\end{lemma}
\begin{proof}
Due to \Cref{lemma:canonical_containment}, $v_1'$ is admissible according to $c_1^*$ (as $c_1 \sqsupseteq c_1^*$; see \Cref{table:notation_reduction}).
As $v_0'$ is not admissible according to $c_1^*$, $v_0' \neq v_1'$.
\end{proof}

We are ready to prove the correctness of our reduction (\Cref{algorithm:reduction}).

\begin{lemma}
\Cref{algorithm:reduction} is a correct weak consensus algorithm with the same message complexity as $\mathcal{A}$.
\end{lemma}
\begin{proof}
\emph{Termination} and \emph{Agreement} of \Cref{algorithm:reduction} follow directly from the properties of $\mathcal{A}$.
Moreover, the message complexity of \Cref{algorithm:reduction} is identical to that of $\mathcal{A}$ as \Cref{algorithm:reduction} introduces no additional communication.
Thus, it is left to prove \emph{Weak Validity}, which we do by analyzing two fully correct executions in which all processes propose the same value:
\begin{compactitem}
    \item Assume that every process is correct and proposes $0$ to weak consensus (line~\ref{line:wbc_propose}).
    Therefore, every process proposes its proposal from the input configuration $c_0 \in \mathcal{I}_n$ (of $\mathcal{P}$) to the underlying algorithm $\mathcal{A}$ (line~\ref{line:wbc_propose_c1}).
    As every fully correct execution is completely determined by the proposal of processes, $\mathcal{A}$ exhibits the execution $\mathcal{E}_0$ which decides $v_0'$.
    Thus, all processes decide $0$ from weak consensus (line~\ref{line:wbc_decide_0}).

    \item Assume that every process is correct and proposes $1$ to weak consensus (line~\ref{line:wbc_propose}).
    Hence, every process proposes its proposal from the input configuration $c_1 \in \mathcal{I}_n$ (of $\mathcal{P}$) to the underlying algorithm $\mathcal{A}$ (line~\ref{line:wbc_propose_c2}).
    Hence, $\mathcal{A}$ exhibits the execution $\mathcal{E}_1$ which decides $v_1'$.
    By \Cref{lemma:reduction_helper}, $v_1' \neq v_0'$, which implies that all processes decide $1$ from weak consensus (line~\ref{line:wbc_decide_1}).
\end{compactitem}
As \emph{Termination}, \emph{Agreement} and \emph{Weak Validity} are proven to be satisfied, the lemma holds.
\end{proof}